%% file: SwitchingRandomisation.tex
\documentclass[11pt, 3p]{elsarticle}
\makeatletter
\def\ps@pprintTitle{%
 \let\@oddhead\@empty
 \let\@evenhead\@empty
 \def\@oddfoot{\centerline{\thepage}}%
 \let\@evenfoot\@oddfoot}
\makeatother

\date{}
\def\texpsfig#1#2#3{\vbox{\kern #3\hbox{\includegraphics{#1}\kern #2}}\typeout{(#1)}}

\usepackage{url}
\usepackage[hidelinks]{hyperref}
\usepackage{bbm}
\usepackage{eurosym}                           
\usepackage[latin1]{inputenc}                  
\usepackage{url}                               
\usepackage{longtable}                         
\usepackage{array}                             
\usepackage{mathrsfs}

\usepackage{amsthm}
\usepackage{amsbsy}

\usepackage{placeins}  

\usepackage{amssymb}
\usepackage{amsmath}
\usepackage{bm}  

\usepackage{enumerate}

\usepackage{graphicx,color}
\usepackage{epsfig}
\usepackage[ruled,vlined]{algorithm2e}
\usepackage{multirow,bigdelim}
\usepackage[table]{xcolor}

\usepackage{natbib}
\setcitestyle{authoryear,open={(},close={)}}

\usepackage{cleveref}

\usepackage{diagbox} 

\usepackage{mathtools,eqparbox}

\usepackage{tikz}

\theoremstyle{plain}
\newtheorem{thm}{Theorem}[section]

\newtheorem{dfn}[thm]{Definition}

\newtheorem{rem}{Remark}[section]

\theoremstyle{remark}

\theoremstyle{plain}
\newtheorem{lem}[thm]{Lemma}

\newtheorem{prop}[thm]{Proposition}
\newtheorem{coroll}[thm]{Corollary}

\theoremstyle{definition}

\newcommand{\e}{{\rm e}}        
\def\R{\mathbb{ R}}             
\def\E{\mathbb{ E}}             
\def\Q{\mathbb{ Q}}             

\def\N{\mathbb{N}}

\def\P{\mathbb{ P}}             

\def\F{\mathcal{F}}             
\def\G{\mathcal{G}}
\def\Cov{\mathrm{\mathbb{C}ov}}

\def\Var{\mathrm{\mathbb{V}ar}}   
\renewcommand{\d}{{\,\rm d}}      

\def\e{{\mathrm{e}}}
\def\1{{\mathbbm{1}}}            

\theoremstyle{plain}

\usepackage[margin=1cm]{caption}

\numberwithin{equation}{section}	     

\let\originalleft\left
\let\originalright\right
\renewcommand{\left}{\mathopen{}\mathclose\bgroup\originalleft}
\renewcommand{\right}{\aftergroup\egroup\originalright}
\newcommand{\X}{X^{\bm{\vartheta}}}  	
\newcommand{\Xreal}{X^{\bm{\theta}}}  	
\newcommand{\fXquad}{f\left(x; \Xquad(t)\right)}  	 
\newcommand{\Xzeta}{X^{\bm\zeta, \bm\vartheta}}  	
\newcommand{\Xzetareal}{X^{\bm z, \bm\vartheta}}  	
\newcommand{\Xlv}{\bar X}  
\newcommand{\Xstoch}{\widehat X}  
\newcommand{\Xfullstoch}{\widetilde X}  
\newcommand{\Y}[1]{Y^{\vartheta}_{#1}}  	
\newcommand{\Yreal}[1]{Y^{\theta}_{#1}}  	
\newcommand{\Yquad}[2]{Y^{\theta_{{#1}_{#2}}}_{#2}}  
\newcommand{\Wquad}{\Bigl(\prod\limits_{j=0}^M w_{n_j}\Bigr)}  	
\newcommand{\Vquad}{V_{\smallvert\bm\ell\smallvert}}
\newcommand{\sumquad}{\sum\limits_{{n_0, \dots, n_M}{=1}}^{N_0, \dots, N_M}}  	
\newcommand{\sumquadL}{\sum\limits_{\smallvert\ell\smallvert = \bm 1}^{\smallvert \bm L\smallvert}}
\newcommand{\smallvert}{\scalebox{.5}{\raisebox{.4pt}{$\vert$}}}   	 
\newcommand{\thetaquad}{{\bm \theta}_{\smallvert \bm{n} \smallvert}}  	
\newcommand{\zquad}{{\bm z}_{\smallvert \bm{\ell} \smallvert}}  	
\newcommand{\Xquad}{X^{\thetaquad}}  	
\newcommand{\Xquadstoch}{X^{\zquad, \thetaquad}}  	 
\newcommand{\Xzetaquad}{X^{\zquad, \bm\vartheta}}  	 
\newcommand{\Xzetaquadquad}{X^{\zquad, \thetaquad}}  	
\newcommand{\XMarkov}{X(t; \bm{\vartheta}, R)}  
\newcommand{\psireal}[1]{\psi^{\theta}_{#1}}  	
\newcommand{\Xfull}{X^{\zeta, \vartheta, \mathcal{M}}}
\title{{Consistent asset modelling with random coefficients and switches between regimes}}

\begin{document}

\author[1]{Felix L.\ Wolf}\
\ead{Felix.Wolf@ulb.be}
\author[1]{Griselda Deelstra}
\ead{Griselda.Deelstra@ulb.be}
\author[2,3]{Lech A.\ Grzelak}
\ead{L.A.Grzelak@uu.nl}
\address[1]{Department of Mathematics, Universit\'e libre de Bruxelles, Brussels, Belgium}
\address[2]{Mathematical Institute, Utrecht University, Utrecht, the Netherlands}
\address[3]{Rabobank, Utrecht, the Netherlands}

\footnotesize{
\begin{abstract}
\noindent
{We explore a stochastic model that enables capturing external influences in two specific ways. The model allows for the expression of uncertainty in the parametrisation of the stochastic dynamics and incorporates patterns to account for different behaviours across various times or regimes.}
To establish our framework, we initially construct a model with random parameters, where the switching between regimes can be dictated either by random variables or deterministically. 
{Such a model is highly interpretable.}
{We further ensure mathematical consistency by demonstrating that the framework can be elegantly expressed through \emph{local volatility} models taking the form of standard jump diffusions.}
Additionally, we consider a Markov-modulated approach for the switching between regimes characterised by random parameters.
For all considered models, we derive characteristic functions, providing a versatile tool with wide-ranging applications. 
In a numerical experiment, we apply the framework to the financial problem of option pricing. The impact of parameter uncertainty is analysed in a two-regime model, where the asset process switches between periods of high and low volatility imbued with high and low uncertainty, respectively.
\end{abstract}
}
\maketitle


{\let\thefootnote\relax\footnotetext{The views expressed in this paper are the personal views of the authors and do not necessarily reflect the views or policies of their current or past employers.}}
\normalsize
%
\section{Introduction}\label{sec:intro}
%
\begin{figure}[h]
\centering
\includegraphics[width=.5\textwidth]{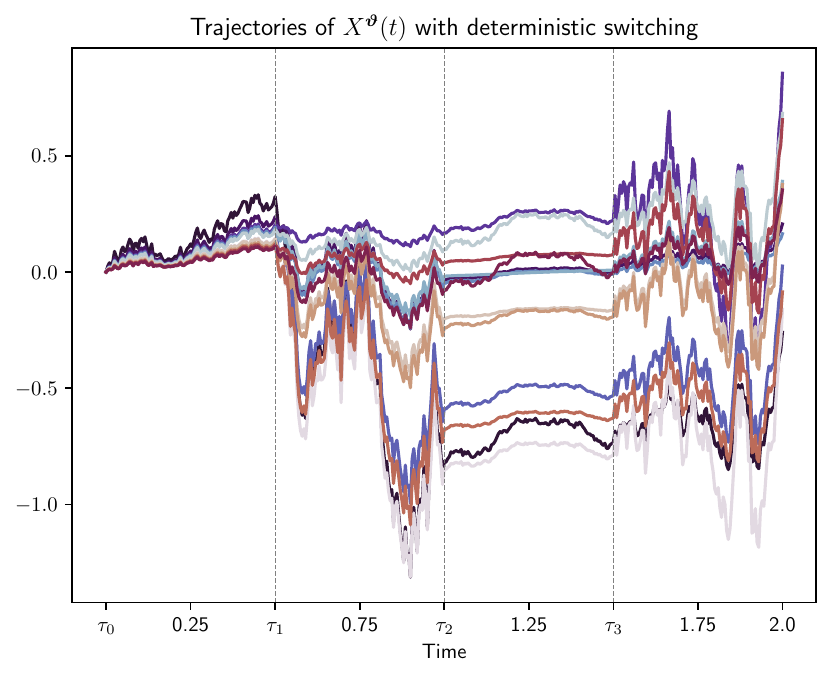}
\caption[Path simulation: Deterministic switching, randomised]
{The figure illustrates trajectories of the composite process $\X(t)$ with deterministic switches and randomised volatility. The regimes are associated with randomiser distributions alternating between a low and a high mean.
Depicted are the trajectories resulting from a single underlying Brownian motion path and several samples of the randomisers. }
\label{fig1}
\end{figure}
{In various domains, stochastic processes are used to account for the unforeseeable nature of the modelled subject.} 
In this article, we describe a novel approach for a jump diffusion 
{that offers additional flexibility in the model parametrisation and a shifting mechanism between regimes with varying uncertainties and model dynamics.}
{Whilst the framework is presented in a comprehensive manner that can be adapted to many modelling cases, we draw conclusions for the modelling of financial assets at various points throughout the article.} One application is found in risk management, where scenarios involving seasonality can be constructed and quickly priced due to the availability of characteristic functions.
\par
The subject of study is a stochastic process obtained {through multiple steps}. First, for each regime, a \emph{component process} is defined with its particular dynamics. 
Then, \emph{switching times} are defined and serve as the points of concatenation at which these component processes are combined to form a \emph{composite process} exhibiting regime switches. These switching times may be stochastic.
The uncertainty feature is encoded into the component processes, each of which depends on an additional random variable we denote the \emph{randomiser}. Instead of using deterministic drift and volatility coefficients, we extend these to functions that are dependent on the randomiser.
However, this is closely related to mixture approaches which can introduce ambiguities regarding the model definition in problems involving nested expectations, {an example thereof is presented in \cite{PiterbargHangover}.} 
{This issue is addressed in the final step of the modelling approach, which involves constructing another stochastic process that closely mimics the dynamics of the process with randomisers, but no longer contains additional random variables.}
{The resulting final processes belong to what are known as \emph{local volatility models} in the Mathematical Finance literature.}
\par
Multiple models are proposed in this article which vary in how the composite process is obtained from the component processes. Departing from deterministic switching times, we advance to more sophisticated models based on stochastic switching times and a Markov-modulated model. In all cases, characteristic functions of the processes are obtained.
\par
\Cref{fig1} depicts exemplary paths of a simple composite process with deterministic switches and randomised volatility coefficients,
\begin{equation}\label{eq:easyX}
\d \X(t) = \sum_{j=0}^3 \1_{t\in[\tau_j, \tau_{j+1})} \left(\left(0.05 - \frac{\vartheta_j^2}{2}\right) \d t + \vartheta_j \d \widetilde W(t)\right),
\end{equation}
for some driving Brownian motion $\widetilde W(t)$, a random vector $\bm\vartheta = (\vartheta_0, \vartheta_1, \vartheta_2, \vartheta_3)$ and deterministic switches at $\tau_1 = 0.5, \tau_2=1, \tau_3 = 1.5$. 
The process alternates between two types of regimes characterised by low and high volatility, expressed through randomisers $\vartheta_j$, $0\leq j\leq3$, following distributions with low and high means, respectively. Within each regime, the volatility varies between trajectories depending on the sample of the randomiser. This creates distinctions between trajectories, although each depicted trajectory is driven by the same realisation of the underlying Brownian motion $\widetilde W(t)$. A comparison to the trajectories of the final local volatility model is drawn in \Cref{fig2}, presented in the numerical experiment section later in this article.
\par
Besides the increased interpretability obtained from incorporating random variables into the model parametrisation, the approach also results in a richer process structure. 
{For example, \cite{GrzelakRand1} demonstrated that this feature enables affine diffusions to accurately capture the volatility smiles observed in the market data of Financial options.}
The construction of a local volatility model {to address the issues of mixtures}, was initially introduced by \cite{Brigo2000risk} and more recently revisited in the context of affine interest rate models by \cite{GrzelakRand2}. 
The challenge in transitioning to a local volatility model lies in identifying the appropriate local volatility structure corresponding to the mixture distribution. 
In our approach, an ansatz is found utilising the Gauss quadrature for the measure integrals corresponding to the mixture distribution.
To illustrate the concept, consider the above composite process $\X(t)$ on the interval $[0, \tau_1)$, which we here denote $Y^{\vartheta_0}(t)$. For every $t\in[0, \tau_1)$, its probability density function is given by
\begin{equation}
f\left(x;Y^{\vartheta_0}(t)\right) = \int_{D_0} f\left(x; Y^{\theta}(t)\right)\d \theta,
\end{equation}
where $D_0$ denotes the domain of the random variable $\vartheta_0$ and $\d Y^{\theta}(t) = (r-\theta^2/2)\d t + \theta \d \widetilde W(t)$ is a standard diffusion with deterministic volatility coefficient $\theta\in D_0$. For a large class of continuous random variables $\vartheta_0$, this integral can be represented by
\begin{equation}\label{eq:introquad}
\int_{D_0} f\left(x; Y^{\theta}(t)\right)\d \theta = \sum_{i=1}^N w_i f(x; Y^{\theta_i}(t)) + \varepsilon_N,
\end{equation}
where $(w_i, \theta_i)_{i=1}^N$ are the {Gauss quadrature pairs} associated with the discretisation and $\varepsilon_N$ is an exponentially vanishing approximation error. Computation of these quadrature pairs can be undertaken efficiently and solely based on moments of the distribution of $\vartheta_0$ through the algorithm of \cite{golub1969calculation}. 
The right-hand side of \eqref{eq:introquad} serves as an arbitrarily close approximation of the density of $Y^{\vartheta_0}(t)$ in which no additional random variables appear and we can adapt the approach of \cite{Brigo2000risk} to construct a local volatility model.
\par
{Switches between process dynamics across different regimes allow for the modelling of seasonal behaviours such as business cycles.}
{Our approach emphasises the stochastic modelling of the switching times $\tau_j$, introduced in \eqref{eq:easyX}. This enables the utilisation of the previously discussed quadrature method, as it transforms the times spent in each regime into another set of continuous random variables within the process structure. Subsequently, we obtain another local volatility model in which the marginal density coincides with an arbitrarily exact approximation of the model involving random parameters and random switching times.}
{A more traditional approach of regime-switching involves the modelling of an underlying Markov process that determines the state of the process.} This goes back to at least \cite{Masi1995}, who consider the hedging problem of a European option following a Markov-modulated underlying. A broad overview of Markov-modulated regime-switching models is provided by \cite{Elliott2005}. 
{We also introduce a Markov-modulated randomised model, which provides a connection to this broad field of study.}

\par
The article is structured as follows. 
In \Cref{sec:randomisationSetting}, we establish the randomisation setup and fully construct the composite process using deterministic switching times, initially introduced in \eqref{eq:easyX} as a simplified toy model.
In \Cref{sec:SDEdet}, the local volatility model which circumvents potential issues with the randomisation formulation is constructed. We obtain error bounds for its density approximation compared to the randomised model as well as its characteristic function. 
The extension to regime switches at stochastic times is the subject of \Cref{sec:StochSwitch}. After enhancing the underlying probabilistic framework to allow for the stochastic switching times, we follow the previously established procedures of \Cref{sec:randomisationSetting,sec:SDEdet} in constructing composite processes, transforming these into local volatility models and obtaining their characteristic function. Here, we distinguish between two types of stochastic switching, involving a fixed and a random number of switches between regimes. 
In \Cref{sec:MMRP}, we propose a Markov-modulated randomised framework in which the regime switches are driven by an underlying Markov chain and obtain the characteristic function of the underlying process. 
Numerical results are showcased in \Cref{sec:numerics}, featuring trajectories of both the local volatility and stochastic switching models. Furthermore, we illustrate a financial application by solving the pricing problem of a European option with an underlying that is modelled using the proposed local volatility models. The results are summarised in \Cref{sec:conclusio} and additional proofs are given in the appendix.

%
\section{Composite randomisation setting}\label{sec:randomisationSetting}
In this section, we define the framework of random coefficients and switches between regimes. Randomness in the coefficients leads to what we denote as \emph{randomised processes}, which are stochastic processes that incorporate an additional source of randomness, influencing both their volatility and drift coefficients. For each path of a randomised process, a sample is drawn from a random variable immediately after the initial time, which henceforth influences the dynamics of this path.
The second main feature, switches between different regimes, is implemented by a form of concatenation of stochastic processes. We define a family of randomised processes with different coefficients, labelling each one as a \emph{component process}. Then, a composition rule is defined from which a \emph{composite process} emerges, which exhibits switching behaviour whenever its constituting component process changes.
\par
We assume the existence of a filtered probability space $(\bar\Omega, (\bar\F_t)_{t\geq0}, {\bar\Q})$ that satisfies the usual conditions of right continuity and completeness and is rich enough to support Brownian motion and jump processes.
On this probability space, stochastic processes in the classical sense, without an additional layer of randomisation, can be defined. 
\par
Further, we consider a probability space $(\Omega^*, \mathcal{A}, \Q^*)$ on which we define a continuous, real-valued random vector 
	\begin{equation}
	\bm\vartheta = (\vartheta_0, \dots, \vartheta_M) \in \R^{M+1},
	\end{equation}
with independent components $\vartheta_j$, $j\in\{0,\dots, M\}$ for some number $M\in\N$. We denote these random variables the \emph{randomisers}. 
\par
This construction, coupled with the assumption of independence between stochastic drivers and randomisers, establishes a valuable link between the randomised processes, with coefficients influenced by random variables $\vartheta_j$, and what we denote their associated conditional processes. In these conditional processes, the coefficients are contingent on real numbers $\theta_j$, which we interpret as realizations of the randomisers, $\theta_j = \vartheta_j(\omega^*)$ for realisations $\omega^* \in \Omega^*$. Since the conditional processes lack the additional layer of randomness, conventional results for stochastic processes can be applied to illuminate the characteristics of the randomised processes.
%
We thus define the probability space on which the randomised processes are defined by $(\Omega, (\G_t)_{t\geq 0}, \Q)$, where $\Omega := \bar\Omega \times \Omega^*,$ $\Q = \bar\Q \otimes \Q^*$, and $\G_{t^+} = \sigma(\F_{t^+} \cup \mathcal{A})$.
\par
Given a random vector of randomisers $\vartheta = \{\vartheta_0, \dots, \vartheta_M\}$, we formally define the \emph{randomised component processes} $\Y{j}(t)$, $j\in\{0,\dots,M\}$ {in \Cref{dfn:componentprocess}}. 
These processes are jump diffusions in which the drift and volatility terms depend on time and the random variable $\vartheta_j$. {We designate these processes as `components' since each one determines the dynamics in a specific regime of the later defined \emph{composite} process.}
In addition, we define the \emph{conditional component processes} $\Yreal{j}(t)$ as the stochastic processes obtained when a realisation $\theta_j = \vartheta_j(\omega^*)$ is given, hence these processes do not exhibit randomness in their coefficients and can be used to explore characteristics of the randomised component processes. 
	\begin{dfn}[Component processes]\label{dfn:componentprocess}
	{For every $j\in\{0, \dots, M\}$ and $t\geq0$, let $W_j(t)$ be a standard Brownian motion and let $P_j(t)$ be a Poisson process with intensity $\lambda \geq 0$. 	Let $\eta_j$ be a real-valued random variable with cumulative distribution function $F_\eta$. We assume mutual independence between all these sources of randomness. Further, let $b_j(t, z)$ and $\sigma_j(t,z)$, respectively, be real-valued functions which are finite for all $t\geq0$ and bounded for all $z\in\R$.
	\par
	Given the real-valued random variable $\vartheta_j$ previously defined as the randomiser, we introduce the \emph{randomised component process} $\Y{j}(t)$, $t \geq 0$, on the probability space $(\Omega, \G_t, \Q)$. It is given by
		\begin{align}
		\Y{j}(t) := \int_0^t b_j(u, \vartheta_j) \d u + \int_0^t \sigma_j(u, \vartheta_j)\d W_j(u) + \sum\limits_{i=1}^{P_j(t)} \eta_j(i), \label{eq:Y}
		\end{align}
	where $\eta_j(i)$ represents the magnitude of the $i$th jump of $P_j(t)$.
	\par
	For any realisation $\theta_j = \vartheta_j(\omega^*)$, we introduce the \emph{conditional component process} $\Yreal{j}(t)$, $t\geq0$, which is defined on the probability space $(\bar\Omega, \bar\F_t, \bar\Q)$ by
		\begin{equation}\label{eq:Yreal}
		\Yreal{j}(t) := \int_0^t b_j(u, \theta_j) \d u + \int_0^t \sigma_j(u, \theta_j)\d W_j(u) + \sum\limits_{i=1}^{P_j(t)} \eta_j(i).
		\end{equation}
	}
	\end{dfn}
Given the independence between components $\vartheta_j$ and $\vartheta_k$, for $j\neq k$, it follows immediately that any two component processes are independent of one another. {Note that the randomised component processes $\Y{j}(t)$ depend only on the random variable $\vartheta_j$, not the entire random vector $\bm\vartheta$. The notation is to be understood as a short-hand for $Y^{\vartheta_j}_j(t)$, the same holds for the conditional component process $\Yreal{j}(t)$ and its characterising realisation $\theta_j$.}
\par
Each conditional component processes \eqref{eq:Yreal} is a jump diffusion suitable to fit various popular financial models, {see \Cref{sec:numerics} in which we consider the Merton jump diffusion for log-prices \citep{Merton1976}.}
The randomised component process \eqref{eq:Y} can be understood as an extension in which the drift and volatility coefficients depend on a random variable $\vartheta_j$ whose outcome is immediately known after the initial time, analogous to the model introduced by \cite{GrzelakRand1}. The interpretation of the parameter $\theta_j$ as a realisation of the random variable $\vartheta_j$ forms the connection by which randomised processes can be treated analytically.
\par
{To formalise the randomised composite process $\X(t)$, we introduce the \emph{switching times} $0=\tau_0 < \tau_1 < \dots < \tau_M < \infty$, which form a partition of the time horizon of the process. For now, we assume the switching times to be deterministic. In \Cref{sec:StochSwitch}, stochastic switching times are studied.
In informal terms, the randomised composite process initiates at $\X(0)\in\R$ and follows the behaviour of $\Y{0}$ until $\tau_1$, then transitions to $\Y{1}$ dynamics until $\tau_2$, and so forth.}
Formally, we define a time shift which is used to appropriately concatenate the component processes. 
For every $j\in\{0,\dots,M\}$, the time shift $s_j(t)$ is defined by
	\begin{equation}\label{def:sj}
	s_j(t) :=
		\begin{cases}
		0, & t < \tau_j, \\
		t - \tau_j, & \tau_j \leq t < \tau_{j+1}, \\
		\tau_{j+1} - \tau_j, & \tau_{j+1} \leq t.	
		\end{cases}
	\end{equation}
Note that $s_M(t) = t - \tau_M$ for all $t \geq \tau_M$. {By summing over time-shifted component processes $\Y{j}(s_j(t))$, $j\in\{0, \dots, M\}$, we obtain the randomised composite process $\X(t)$ which follows the dynamics of $\Y{j}$ on the interval $[\tau_j, \tau_{j+1})$.}

	\begin{dfn}[Composite process]\label{def:Xprocess}
	For every $j\in\{0,\dots,M\}$, let $\Y{j}$ be the randomised component process given in \Cref{dfn:componentprocess}. 
	The \emph{randomised composite process} $\X(t)$, $t\geq 0$, is given by
		\begin{equation}\label{eq:generalX}
		\X(t) = x_0 + \sum_{j=0}^M \Y{j}(s_j(t)),
		\end{equation}
	with some initial value $x_0 \in \R$.
	\par
	By conditioning on a realisation $\bm\theta = (\theta_0, \dots, \theta_M) = (\vartheta_0(\omega^*), \dots, \vartheta_M(\omega^*)) = \bm\vartheta(\omega^*)$, we obtain the \emph{conditional composite process} $\Xreal(t)$, defined on the probability space $(\bar\Omega, (\bar\F_t)_{t\geq0}, \bar \P)$. It is given by
		\begin{equation}\label{eq:Xrealtoo}
		\Xreal(t) = x_0 + \sum_{j=0}^M \Yreal{j}(s_j(t)).
		\end{equation}
	\end{dfn}
In the following result, we represent the conditional composite process $\Xreal(t)$ directly as a jump-diffusion SDE.
	\begin{prop}[The conditional composite SDE]\label{prop:XrealSDE}
	Let $\bm \theta = (\theta_0, \dots, \theta_M) \in \R^{M+1}$ be a constant vector and let $\Xreal(t)$ be the associated conditional composite process defined in \Cref{def:Xprocess}. This process is the solution of the SDE
	\begin{equation}\label{eq:XrealSDE}
	\d \Xreal(t) = \beta(t;\bm\theta)\d t + \gamma(t;\bm\theta) \d \widetilde W(t) + \d \sum_{i=1}^{\widetilde P(t)}\pi(i) , 
	\quad \Xreal(0)=x_0,
	\end{equation}
	where $\widetilde W(t)$ is a standard Brownian motion, $\widetilde P(t)$ a Poisson process with intensity $\lambda\geq0$ and $\pi(i)$ is the (i.i.d.)\ magnitude of the $i$th jump, following a distribution with cdf $F_\eta$.
	The drift and volatility coefficients are given by
	\begin{equation}\label{eq:betagammadef}
	\beta(t; \bm\theta) := \sum_{j=0}^M b_j(s_j(t), \theta_j) \1_{t \in [\tau_j, \tau_{j+1})},
	\quad
	\gamma(t; \bm\theta) := \sum_{j=0}^M \sigma_j(s_j(t), \theta_j) \1_{t \in [\tau_j, \tau_{j+1})}.
	\end{equation}
	\end{prop}
This result follows directly from the \Cref{def:Xprocess} and its proof is given in \mbox{\ref{appx:XrealSDEproof}.}
\par
{In the given forms, the randomised processes $\X(t)$ and $\Y{j}(t)$, $j\in\{0,\dots,M\}$ are difficult to treat, whereas the conditional processes $\Xreal(t)$ and $\Yreal{j}(t)$ are {more} tangible. 
In \cite{GrzelakRand1} it is shown how certain features of the randomised processes, such as the probability density function (pdf) and the characteristic function (chf), can be obtained by integrating the known results for the conditional process against the probability measure induced by the respective randomisers.}
\par

{We exemplarily give the characteristic function of the randomised component process $\Y{j}(t)$ to illustrate this principle.}
	\begin{lem}[Characteristic function of the randomised component process]\label{lem:Ychf}
	For every $j\in\{0, \dots, M\}$ and time $t\geq 0$, let the characteristic function of $\Yreal{j}(t)$ be denoted by
		\begin{equation}\label{eq:Yrealchf}
		\varphi(u; \Yreal{j}(t)) := \E_0\left[\exp\left(iu\Yreal{j}(t)\right)\right].
		\end{equation}
	Then, the characteristic function of $\Y{j}(t)$ is given by
		\begin{align}
		\varphi(u; \Y{j}(t))
		:= \E \left[\exp\left(iu\Y{j}(t)\right) \right] 
		&= \int_{D_j} \varphi(u; Y^{\theta_j}_j(t)) \d F_{\vartheta_j}(\theta_j) \nonumber \\
		& {= \int_{D_j} \varphi(u; Y^{\theta_j}_j(t)) f_{\vartheta_j}(\theta_j) \d \theta_j,}
		\label{eq:Ychfgeneral1}
		\end{align}
	where $D_j$ denotes the domain of the random variable $\vartheta_j$, and $F_{\vartheta_j},\ {f_{\vartheta_j}(y)}$ its cumulative distribution function {and probability density function, respectively.}
	\end{lem}
	\begin{proof}
	The result follows immediately from the tower property of expectations,
		\begin{equation}
		\E \left[\exp\left(iu\Y{j}(t)\right) \right] 
		= \E \left[\E\left[\left.\exp\left(iu\Y{j}(t)\right) \right| \vartheta = \theta \right] \right]
		= \E \left[ \varphi(u; \Yreal{j}(t)) \right].
		\end{equation}
	\end{proof}
{A fast numerical approximation technique with a quantifiable error is provided by Gauss quadrature. This method results in a discretisation of the integral for which the algorithm of \cite{golub1969calculation} can compute the necessary weights and points based solely on the moments of the randomiser's probability distribution.}
	\begin{lem}[Discretisation of the characteristic function]\label{lem:Ychfquad}
		{For $j\in\{0,\dots,M\}$ let $N_j \in \N$ be the order of approximation. If the moments of the randomiser $\vartheta_j$ are finite for every $n_j \leq 2N_j$, $\E[\vartheta_j^{n_j}] < \infty$, then} the characteristic function {$\varphi(u; \Y{j}(t))$ is represented by the discretisation}
		\begin{equation}\label{eq:Ychfgeneral2}
		\varphi(u; \Y{j}(t))
		= \int_{D_j} f_{\vartheta_j}(\theta_j) \varphi(u; Y^{\theta_j}_j(t)) \d \theta_j 
		= \sum\limits_{n_j=1}^{N_j} w_{n_j}\,  \varphi(u; Y^{\theta_{n_j}}_j(t)) + \varepsilon_{N_j}(t, u). 
		\end{equation}
	Here, $(w_{n_j}, \theta_{n_j})_{n_j=1}^{N_j}$ are the Gauss-quadrature pairs associated with integration against the weight function $f_{\vartheta_j}(\theta_j)$, and $\varepsilon_{N_j}(t, u)$ is the quadrature approximation error, which is bounded by 
		\begin{equation}\label{eq:Ychfquaderror}
		\varepsilon_{N_j}(t, u) \leq \sup\limits_{\xi \in D_j} \frac{1}{(2N)!} \left.\frac{\partial^{2N}}{\partial \theta^{2N}} \varphi(u; Y^{\theta_j}_j(t)) \right|_{\theta = \xi} .
		\end{equation}
	\end{lem}
The proof of this quadrature discretisation result follows from an application of Theorem~2.1 in \cite{GrzelakRand1}. {In \ref{appx:RandomisedPDF}, we analogously obtain the pdf of the randomised component process $\Y{j}(t)$ in this way.} 
\par
{The discretization achieved through Gauss quadrature is significant for the applicability of this technique, as it links the random composite process to a finite number of concrete conditional component processes. This is valuable in applications like the process calibration, where each conditional process may benefit from an analytical form allowing rapid calibration.}
\par
{However, a limitation of this approach arises when dealing with nested expectation-type problems. Sampling of the randomisers immediately after the initial time complicates the inclusion of assets whose pricing model requires nested expectations with an inner expectation conditioned on a time after the initial one, as \cite{PiterbargHangover} observed.}
{This issue can be avoided through an alternative approach proposed by \cite{Brigo2000risk}. 
In this approach, a stochastic process $\Xlv(t)$ is defined that shares the marginal distributions of the randomised composite process $\X(t)$, but can be defined on the probability space $(\bar\Omega, \bar\F_t, \bar\Q)$ which is devoid of any additional randomisation features. 
The following section is dedicated to this formulation.}
%
\section{The local volatility model}\label{sec:SDEdet}
%
{The randomised model provides a versatile approach to modelling, is easily understandable, and incorporates a readily accessible discretisation scheme using the quadrature method. 
To address the limitation highlighted earlier, we now introduce a local volatility model for the composite process that is fully defined within the well-established framework of stochastic processes with deterministic parameters, while maintaining the marginal distributions obtained from the quadrature discretization of the randomised composite process. 
The results of this section are obtained under the assumption of deterministic switching times, the scenario involving stochastic switching times is treated in \Cref{sec:StochSwitch}.}
\par
This section is structured as follows. We first give the main result, which is the SDE of a local volatility process with jumps, denoted $\Xlv(t)$, 
{that has no added layer of randomisation. We show that the solution of this SDE has a probability density function of the same shape as is obtained from the randomised composite process $\X(t)$, which implies a parametrisation of the SDE such that the marginal densities of the local volatility model and the randomised composite process coincide up to the quadrature discretisation error of the randomised model.}
	\begin{thm}[Local volatility formulation of deterministic switching]\label{thm:detswitchSDEthm}
	Let $\theta_{n_j}\in\R$ and $w_{n_j}\in\R_+$, with $n_j \in \{1, \dots, N_j\}$, $N_j \in \N$, $j\in\{0, \dots, M\}$ be a collection of constants such that $\sum_{n_j=1}^{N_j} w_{n_j} = 1$. 
	{For notational convenience, we denote 
	$\thetaquad := (\theta_{n_0}, \dots, \theta_{n_M})$.
	}
	Consider the jump-diffusion SDE
	\begin{equation}\label{eq:XlocalvolSDE}
	\d\Xlv(t) = \bar\mu\left(t, \Xlv(t)\right) \d t + \bar\sigma\left(t, \Xlv(t)\right) \d \bar W(t) + {\d \sum_{i=1}^{\bar N(t)}\bar\eta(t)}, \quad \Xlv(0) = x_0,
	\end{equation}
	where $\bar W(t)$ is a standard Brownian motion, $\bar N(t)$ a Poisson process with intensity $\bar\lambda$, and $\bar \eta(t)$ denote the jump sizes which follow a distribution with cdf $F_{\bar\eta}$ and arrive according to the arrival times of the process $\bar N(t)$.
	Let the jump intensity be given by $\bar\lambda = \lambda$, the jump size distribution have cdf $F_{\bar\eta} = F_{\eta}$, 
	the drift coefficient and volatility coefficient be given, respectively, by
	\begin{align}\label{eq:musolution}
	\bar\mu(t,x) &= \frac{\sumquad \Wquad \beta(t; \thetaquad) \fXquad}{\sumquad \Wquad \fXquad}, 
	\\
	\label{eq:sigmasolution}
	\bar\sigma^2(t, x) &= \frac{\sumquad \Wquad \gamma^2(t; \thetaquad)  \fXquad}{\sumquad \Wquad \fXquad},
	\end{align}
	where $\fXquad$ {denotes} the density of the process $\Xquad(t)$, 
	{which is defined together with} the functions $\beta(t; \thetaquad)$ and $\gamma(t; \thetaquad)$ in \Cref{prop:XrealSDE}.
	\par
	Then the SDE \eqref{eq:XlocalvolSDE} has a {unique}, strong solution $\Xlv(t)$ with probability density function
	\begin{equation}\label{eq:fXlv}
	f(x; \Xlv(t)) = \sumquad \Wquad \fXquad.
	\end{equation}
	\end{thm}
The proof of \Cref{thm:detswitchSDEthm} draws from an argument of identifying Fokker--Planck equations, which has been brought forth in similar form in \cite{Brigo2000risk} and \cite{GrzelakRand2}. The SDE \eqref{eq:XlocalvolSDE} implies a certain time-evolution of the density of $\Xlv(t)$, which can be related to the time-evolution of the density $\fXquad$ through the structure imposed onto the density $f(x; \Xlv(t))$ in \eqref{eq:fXlv}. Since the Fokker--Planck equation (FPE) for $\fXquad$ is explicitly known, its connection to the FPE of $f(x; \Xlv(t))$ then implies the specific parameter choices of $\bar\mu(t, x)$, $\bar\lambda$, $\bar\eta$ and $\bar\sigma^2(t,x)$.
\begin{proof}[Proof of \Cref{thm:detswitchSDEthm}]
Assume that a solution to the SDE in \eqref{eq:XlocalvolSDE} exists. Then, the SDE admits a Fokker--Planck equation for $\Xlv(t)$ (see, for example, Theorem~7.5 in \cite{Hanson2007} for a detailed treatise), given by
\begin{multline}\label{eq:XlvFPEformal}
\partial_t f_{\Xlv(t)} (x) 
= \frac12 \partial^2_{x^2} \left(\bar\sigma^2(t, x) f_{\Xlv(t)} (x)\right)
\\
- \partial_{x} \left( \bar\mu(t, x) f_{\Xlv(t)} (x) \right)
- \bar\lambda f_{\Xlv(t)} (x) 
+ \bar\lambda \int_{D_{\bar\eta}} f_{\bar\eta}(z) f_{\Xlv(t)} (x-z)\d z,
\end{multline}
where $f_{\bar\eta}$ is the density of the jump size distribution and $D_{\bar\eta}$ is its domain.
If the pdf of $\Xlv(t)$ is to be of the shape given in \eqref{eq:fXlv}, linearity implies that
\begin{align}\label{eq:linearitytrickRepeat}
\partial_t f(x; \Xlv(t))
&= \partial_t \left(\sumquad \Wquad \fXquad\right) \nonumber
\\
&= \sumquad \Wquad \partial_t \fXquad. 
\end{align}
The coefficients of $\Xquad(t)$ given in \Cref{prop:XrealSDE} only depend on the time $t$ but not the state $x$ of the process, and thus its FPE is explicitly given by
\begin{multline}\label{eq:XquadFPE}
\partial_t \fXquad 
= \frac12 \gamma^2(t; \thetaquad) \partial^2_{x^2} \fXquad
- \beta(t; \thetaquad) \partial_{x} \fXquad 
\\
- \lambda \fXquad 
+ \lambda \int_{D_\eta} f_\eta(z) f(x-z; \Xquad(t))\d z.
\end{multline}
Therefore, \eqref{eq:XlvFPEformal} and \eqref{eq:linearitytrickRepeat} are equal when the following system of equations holds: 
\begin{align}
\label{eq:volcom}
\frac12 \partial^2_{x^2} \left(\bar\sigma^2(t, x) f_{\Xlv(t)} (x)\right) 
&=
\frac12 \sumquad \Wquad \gamma^2(t; \thetaquad) \partial^2_{x^2} \fXquad,
\\
\label{eq:driftcom}
\partial_{x} \left( \bar\mu(t, x) f_{\Xlv(t)} (x) \right)
&= 
\sumquad \Wquad \beta(t; \thetaquad) \partial_{x} \fXquad,
\\[1em]
\label{eq:intenscom}
\bar\lambda f_{\Xlv(t)} (x)
&= 
\sumquad \Wquad \lambda \fXquad ,
\\[1em]
\label{eq:jumpsizecom}
\bar\lambda \int_{D_{\bar\eta}} f_{\bar\eta}(z) f_{\Xlv(t)} (x-z)\d z
&=
\lambda \int_{D_\eta} f_\eta(z) f(x-z;\Xquad(t))\d z.
\end{align}
A solution to \cref{eq:intenscom,eq:jumpsizecom} is immediately found in $\bar\lambda = \lambda$ and $\bar\eta = \eta$. 
Comparison of the remaining differential equations, and substituting the specific shape of $f(x; \Xlv(t))$ given in \eqref{eq:fXlv}, yields
\begin{multline}\label{eq:comparevola}
\bar\sigma^2(t, x) \sumquad \Wquad \fXquad
\\
= 
\sumquad \Wquad \gamma^2(t; \thetaquad) \fXquad
 + C_1(t)x + C_2(t)
\end{multline}
and
\begin{multline}\label{eq:comparedrift}
\bar\mu(t,x)\sumquad \Wquad \fXquad
\\
=
\sumquad \Wquad \beta(t; \thetaquad) \fXquad
+ C_3(t),
\end{multline}
for some time-dependent functions $C_i(t)$, $i=1,2,3$. Since these equations must vanish for $x\to\pm\infty$, one can immediately deduce that $C_i(t) = 0$, $\forall i$. Rearranging \eqref{eq:comparevola} and \eqref{eq:comparedrift} for $\bar\mu(t, x)$ and $\bar\sigma^2(t, x)$, respectively, yields the postulated solutions $\bar\mu(t,x)$ and $\bar\sigma^2(t,x)$ given in \eqref{eq:sigmasolution}.
\par
Finally, we confirm the existence of the {unique} solution $\Xlv(t)$ for all $t\geq 0$ following the conditions given in {Theorem~1.19 in \cite{Oksendal2007}. The local volatility formulation has not altered the jump component, thus preserving the validity of the conditions from the initial formulation. Therefore, we only need to consider the new drift and volatility coefficients for which we show a uniform boundedness criterion.}
It holds that
\begin{multline}
\bar\sigma^2(t, x) 
= \frac{\sumquad \Wquad \gamma^2(t; \thetaquad)  \fXquad}{\sumquad \Wquad \fXquad}
\\
\leq \sup_{\thetaquad} \gamma^2(t, \thetaquad) \frac{\sumquad\Wquad \fXquad}{\sumquad \Wquad \fXquad}
= \sup_{\thetaquad} \gamma^2(t, \thetaquad), 
\end{multline}
where the supremum is taken over all parameters $\theta_{n_0}, \dots, \theta_{n_M},$ with indices $n_j \in \{0, \dots, N_j\}$. 
Analogously, a bound for the drift term is found in $\sup_{\thetaquad} \beta(t, \thetaquad)$. By the definitions of $\beta$ and $\gamma$ and the boundedness assumption on the component process coefficients $b_j$ and $\sigma_j$, these suprema are uniformly bounded in time, 
\begin{equation}
\max\left(\sup_{\thetaquad} \beta(t, \thetaquad), \sup_{\thetaquad} \gamma^2(t, \thetaquad)\right) \leq K \in \R,
\end{equation}
therefore the Lipschitz and linear growth condition of the SDE are satisfied and the {unique} solution exists.

\end{proof}
{In the next result, we show that the probability density function of the local volatility process $\Xlv(t)$ in \eqref{eq:fXlv} is intrinsically linked to the randomised composite process $\X(t)$. To this end, we first obtain the pdf $f(x; \X(t))$ in its exact form, and then derive the quadrature discretisation which coincides with \eqref{eq:fXlv}.
}
	\begin{thm}[Density interpretation]\label{thm:DensityInterpretation}
	Let $\X(t)$ be the randomised composite process given in \Cref{def:Xprocess}.
		\begin{enumerate}[i.]
		\item The probability density function of $\X(t)$ is given by 
			\begin{align}\label{eq:Xdensi}
			 f(x; {\X(t)})
			= \Bigl(\delta_{x_0} \ast f_{\Y{0}(s_0(t))} \ast \cdots \ast f_{\Y{M}(s_M(t))}\Bigr)(x)
			\end{align}
		where $\delta_{x_0}(x)$ is the translated Dirac delta function, which is zero everywhere except at $x_0$, and $f_{\Y{j}(s_j(t))}(y) := f(y; \Y{j}(s_j(t)))$ for $j\in\{0,\dots, M\}$ as given in \Cref{cor:Ydensity}.
		\item
		If, for every $j\in\{0, \dots, M\}$, the random variable $\vartheta_j$ possesses finite first $2N_j\in \N$ moments, $\E[\vartheta_j^{2N_j}] < \infty$, then the quadrature approximation of the density $f_{\X(t)}(x)$ is given by
			\begin{align}
			f(x;{\X(t)})
			&= \sumquad \Wquad f\Bigl(x; x_0 + \sum_{j=0}^M \Yquad{n}{j}(s_j(t))\Bigr) + \bar\varepsilon(t)
			\nonumber \\
			&= \sumquad \Wquad \fXquad + \bar\varepsilon(t)
			\label{eq:quadratureinterpretationSum}
			\end{align}
		where for every $j\in\{0, \dots, M\}$, $(w_{n_j}, \theta_{n_j})_{n_j=1}^{N_j}$ are the Gauss quadrature pairs associated with the random variable $\vartheta_j$, and the error term $\bar\varepsilon(t)$ can be bounded by
			\begin{equation}\label{eq:compositequaderrorbound}
			\bar\varepsilon(t) \leq \sum_{j=0}^M \sup\limits_{y\in\R} \hat\varepsilon_{N_j}(t, y),
			\end{equation}
		with $\hat\varepsilon_{N_j}(t, y)$ the discretisation error {for which a bound is} given in \eqref{eq:Ychfquaderror}.
		\end{enumerate}
	\end{thm}
{The proof of this result is given in \ref{appx:RandomisedPDF}.}
From the theorem, it follows that choosing the parameters $(w_{n_j}, \theta_{n_j})$, $n_j \in \{1, \dots, N_j\}$, $N_j \in \N$, $j\in\{0, \dots, M\}$ in \Cref{thm:detswitchSDEthm} as the quadrature weights and points of the randomiser sequence $(\vartheta_0, \dots, \vartheta_M)$,
{the solution to the SDE \eqref{eq:XlocalvolSDE} corresponds to the quadrature density of the randomised composite process $\X$, up to the discretisation error $\bar\varepsilon(t)$.}
\par
{The characteristic function of the local-volatility model $\Xlv(t)$ immediately follows from the previous results. It will be used later in \Cref{sec:numerics} for asset pricing experiments.
\begin{coroll}\label{cor:Xlvchf}
The characteristic function of $\Xlv(t)$ is given for every $u \in \R$ and $t \geq 0$ by
\begin{equation}
\varphi(u; \Xlv(t)) = \e^{iux_0} \prod_{j=0}^M \sum_{n_j = 1}^{N_j} w_{n_j} \varphi(u; \Yquad{n}{j}(s_j(t)))
\end{equation}
\end{coroll}
\begin{proof}Given the density \eqref{eq:fXlv} of $\Xlv(t)$ in \Cref{thm:detswitchSDEthm}, the result is a consequence of the convolution theorem of (inverse) Fourier transformation, analogous to \eqref{eq:convolutiontheoremstep}.
\end{proof}
}
\par
In summary, within the context of deterministic switching, we have presented a jump-diffusion {with deterministic coefficients, and the parametrisation under which} its pdf mimics that of the randomised composite process $\X(t)$.
\section{Stochastic switching times}\label{sec:StochSwitch}
In the following, we extend the composite randomisation framework by stochastic switching, departing from deterministic switching times in earlier sections.
More specifically, we introduce positive random variables to represent the time durations during which the component processes control the overall composition.
{Within this section, two models involving stochastic switches are introduced.}
At first, we consider the scenario involving a constant number of switches. That is, the timing of each switch is stochastic, while the overall number of switches remains fixed. 
Based on the established quadrature methodology, we derive a local volatility model for this setting and find the characteristic function of the corresponding composite process.
Building upon this, we progress to a fully stochastic switching setup, in which the number of switches may also be random and obtain the characteristic function of the local volatility model associated with this setting.
\par
Previously, the composite model was characterised by deterministic switching times $\tau_j$, $j\in\{1, \dots, M\}$. 
We now choose a slightly different approach in which we represent each \emph{sojourn time}, formerly given by $\zeta_j = \tau_{j}-\tau_{j-1}$ for $j\in\{1, \dots, M\}$, as a random variable and derive the switching times $\tau_j$ from there. 
Each sojourn time represents the period in which a specific component process influences the overall composition.
\par
In terms of the underlying probabilistic model, we augment the probability space of the randomiser $(\Omega^*, \mathcal{A}, \Q^*)$ to additionally support a positive, real-valued random vector $\bm\zeta := (\zeta_1, \dots, \zeta_{M})\in\R_+^M$, where the random variable $\zeta_j$ represents the duration of the sojourn time of the component $\Y{j-1}$, $j\in\{1, \dots, M\}$. With a total of $M$ switches, the sojourn time of the last component $\Y{M}$  is determined implicitly by the preceding sojourn times and the time $t$ at which the model is observed.
\par
We assume that the stochastic sojourn times $\zeta_1, \dots, \zeta_{M}$ are pairwise independent, as well as independent of the (parameter) randomiser $\bm\vartheta$.
Then, the randomiser probability space can be partitioned into a product space with $\Omega^* = \Omega^*_1\times \Omega_2^*$, $\mathcal{A}=\mathcal{A}_1\otimes\mathcal{A}_2$, and $\Q^* = \Q^*_1 \otimes \Q^*_2$. The parameter randomiser $\bm\vartheta$ is defined on $(\Omega^*_1, \mathcal{A}_1, \Q^*_1)$, similarly to the construction in \Cref{sec:randomisationSetting}, and the sojourn times $\bm\zeta$ on $(\Omega^*_2, \mathcal{A}_2, \Q^*_2)$.
\par
Under this modified construction, the randomised component processes $Y_j(t)$ for $j \in \{0, \dots, M\}$ retain their definition as given in \Cref{dfn:componentprocess}.
However, we need to adjust the definition of the randomised composite process \Cref{def:Xprocess}. 
Specifically, the time shifts $s^{\bm \zeta}_j(t)$ are tailored to accommodate the added dependency on the stochastic sojourn times.
\begin{dfn}[The randomised composite process with $M$ stochastic switching times]\label{def:Xprocessstoch}
For some fixed number $M\in\N$, let $\Y{j}(t)$, $j\in\{0, \dots, M\}$ be a collection of randomised component processes and let $\bm\zeta = (\zeta_1, \dots, \zeta_M)$ be the independent, stochastic sojourn times of these components. For every $j\geq 1$, we define the stochastic switching time by
\begin{equation}\label{eq:stochswitchingtimes}
\pi_j := \sum_{k=1}^{j} \zeta_k,
\end{equation}
and set $\pi_0 := 0$. We further define the time shifts $s^{\bm\zeta}_j(t)$ by
\begin{equation}
s^{\bm\zeta}_j(t) := 
\begin{cases}
0, & t < \pi_j,\\
t - \pi_j, & t \in[\pi_j, \pi_{j+1}),\\
\zeta_{j+1}, & t \geq \pi_{j+1},
\end{cases}
\end{equation}
for $j\in\{0, \dots, M-1\}$. For the final component $\Y{M}$, the time shift is defined as $s^{\bm\zeta}_M(t) := (t - \pi_M)\1_{t \geq \pi_M}$.
Under this updated notion of time shifts, we thus define the \emph{randomised composite process with $M$ stochastic switching times} as
\begin{equation}
\Xzeta(t) := x_0 + \sum_{j=0}^M \Y{j}(s^{\bm\zeta}_j(t)).
\end{equation}
\end{dfn}
Notably, when the random vector $\bm\zeta$ is substituted by a realisation $\bm z$, this definition coincides with that of the randomised composite process with deterministic switches given by that realisation.
\par
We formalise this idea through the notion of a \emph{sojourn-time-conditioned} process. 
For a sample $\omega^*_2\in\Omega^*_2$, let $\bm z =(z_1, \dots, z_{M}) := (\zeta_1(\omega^*_2), \dots, \zeta_{M}(\omega^*_2))$ be a realisation of the sojourn times.
For every $j\in\{0, \dots, M\}$, let 
$\tau_j := \pi_j(\omega^*_2) = \sum_{\ell=1}^j z_j$ denote the corresponding realisation of the $j$th switching time.
\par
For every realisation $\bm z$, the associated time shift realisation $s^{\bm z}_j$, $j\in\{0, \dots, M\}$ is given by 
	\begin{equation}\label{eq:stochsj}
	s_j^{\bm z}(t) =
		\begin{cases}
		0, & t < \tau_j \\
		t - \tau_j, & \tau_j \leq t < \tau_{j+1}, \\
		z_j = \tau_{j+1} - \tau_j, 
		& \tau_{j+1} \leq t,
		\end{cases}
	\end{equation}
and $s_M^{\bm z}(t) = t - \tau_M$ for all $t \geq \tau_M$.
This gives rise to a sojourn-time-conditioned version of the composite process, 
	\begin{equation}
	\Xzetareal(t) = x_0 + \sum_{j=0}^M \Y{j}(s_j^{\bm z}(t)).
	\end{equation}
With the sojourn time realisations $\bm z$ prescribing a choice of deterministic switching times $\tau_j$, the thus conditioned process $\Xzetareal(t)$ is a randomised composite process with deterministic switching times exactly as defined in \Cref{def:Xprocess}, with its associated local volatility model given in \Cref{thm:detswitchSDEthm}.
\par
This interpretation of deterministic switching as a realisation of stochastic switching allows us to employ the methodology introduced in earlier sections.
Accordingly, the pdf of $\Xzeta(t)$ is obtained from the integral 
\begin{align}\label{eq:Xzetaquad1}
	f(u; \Xzeta(t)) 
	&:= \int_{\R_{+}^M} f\left(u; \Xzetareal(t)\right) \d F_{\bm\zeta^M}(\bm z) \nonumber \\
	&= \int_{\R_+^M} f_{\bm \zeta^M}(\bm z) f(u; \Xzetareal(t)) \d \bm z,
	\end{align}
where $F_{\bm\zeta^M}$ and $f_{\bm\zeta^M}$ are the cdf and pdf of the joint distribution of sojourn times with the additional condition that there have been exactly $M$ switches by time $t$. This distribution is equivalently given by
\begin{equation}
\bm\zeta^M := (\zeta_1, \dots, \zeta_M) \large\mid \sum_{\ell=1}^{M} \zeta_\ell < t.
\end{equation}
To efficiently apply the quadrature method to \eqref{eq:Xzetaquad1}, we seek to factorise the joint density $f_{\bm \zeta^M}(z)$ into marginal densities because the performance of the quadrature method tends to decline as the dimensionality increases.
\par
We address this with a scheme centred on successive conditioning. The condition $\sum_{\ell=1}^M \zeta_\ell < t$ can be expressed as the intersection over the events 
\begin{equation}\label{eq:succevents}
\Bigl\{\sum_{\ell=1}^M \zeta_\ell < t\Bigr\} = 
\left\{\zeta_0 < t\right\} \cap \left\{\zeta_1 < t - \zeta_0\right\} \cap 
\cdots \cap \left\{\zeta_M < t - \zeta_0 + \dots \zeta_{M-1} \right\}.
\end{equation}
Notably, each event $\mathcal{Z}_j(t) := \{\zeta_j < t - \sum_{\ell = 0}^{j-1} \zeta_\ell \}$ only depends on `preceding' sojourn times $\zeta_\ell$ where $\ell < j$. 
By substituting \eqref{eq:succevents} in the definition of the joint distribution $\bm\zeta^M$, we obtain a decomposition into marginal distributions,
\begin{equation}\label{eq:zetaM}
f_{\bm\zeta^M}(\bm z) = f_{\zeta_1\mid\mathcal{Z}_1(t)}(z_1)  f_{\zeta_2\mid\mathcal{Z}_2(t)}(z_2)\cdots f_{\zeta_M\mid\mathcal{Z}_M(t)}(z_M).
\end{equation}
This form, as we now show, allows us to compute the densities when considered in this particular order, starting with $f_{\zeta_1\mid\mathcal{Z}_1(t)}(z_1)$.
\par
In the following, denote the right-truncated distribution of the continuous random variable $\zeta_j$ in point $b\in\R$ by $R(\zeta_j, b)$. The density of such a truncated distribution is well known to be given by $f_{R(\zeta_j, b)}(z) = f_{\zeta_j}(x) / F_{\zeta_j}(b)$.
By definition, it holds that $\mathcal{Z}_1(t) = \{\zeta_1 < t\}$, and therefore 
\begin{equation}
\zeta_1\mid\mathcal{Z}_1(t) =  R(\zeta_1, t)
\end{equation}
is a right-truncated distribution of which we know the density $f_{\zeta_j}(x) / F_{\zeta_j}(t)$ and thus can compute moments of $\zeta_1\mid\mathcal{Z}_1(t)$. 
From the first $2L_1$ moments of the distribution, the Golub-Welsch algorithm yields the $L_1\in\N$ quadrature pairs $(v_{\ell_1}, z_{\ell_1})_{\ell_1 = 1}^{L_1}$ corresponding to $\zeta_1\mid\mathcal{Z}_1(t)$.
\par
Furthermore, for every quadrature points $z_{\ell_1}$, the resulting distribution $\zeta_2 \mid \zeta_2 < t - z_{\ell_1} = R(\zeta_2, t - z_{\ell_1})$ is in right-truncated form and we can repeat the process to obtain the $L_2\in\N$ quadrature pairs $(v_{\ell_2}^{\ell_1}, z_{\ell_2}^{\ell_1})_{\ell_2 = 1}^{L_2}$ associated with $\zeta_2\mid\mathcal{Z}_2(t)$ and the respective outcome $z_{\ell_1}$ of the first sojourn time. In total, $L_1$ distinct sets of $L_2$ quadrature nodes are obtained.
\par
For each of these quadrature points $z_{\ell_2}^{\ell_1}$, we can repeat the procedure to obtain quadrature pairs $(v_{\ell_3}^{\ell_1, \ell_2}, z_{\ell_3}^{\ell_1, \ell_2})_{\ell_3 = 1}^{L_3}$ associated with the third sojourn time conditioned on $\{\zeta_3 < t - z_{\ell_1} - z_{\ell_2}^{\ell_1}\}$. 
\par
Indeed, for any $j\in\{2, \dots, M\}$ and selected combination of quadrature points $(z_{\ell_1}, z^{\ell_1}_{\ell_2}, \dots, z^{\ell_1,\dots, \ell_{j-2}}_{\ell_{j-1}})$, the procedure yields a right-truncated distribution 
\begin{equation*}
\zeta_{j} \mid \zeta_j < t-\sum_{k=1}^{j-1}z^{\ell_1, \dots, \ell_{k-1}}_{\ell_k}
= R\left(\zeta_j; t-\sum_{k=1}^{j-1}z^{\ell_1, \dots, \ell_{k-1}}_{\ell_k}\right),
\end{equation*}
for which we compute the moments. From these, we derive the associated $L_j\in\N$ quadrature pairs $(v^{\ell_1, \dots, \ell_{j-1}}_{\ell_j}, z^{\ell_1, \dots, \ell_{j-1}}_{\ell_j})_{\ell_j = 1}^{L_j}$.
\par
The method results in $\prod_{j=1}^M L_j$ different sets of quadrature nodes $\{z_{\ell_1}, \dots,  z^{\ell_1, \dots, \ell_{M-1}}_{\ell_M}\}$. Each set defines $M$ realised switching times which are given analogously to the relation \eqref{eq:stochswitchingtimes} by
\begin{equation}
\tau^{\ell_1, \dots, \ell_{j-1}}_{\ell_j} = \sum_{k=1}^{j} z^{\ell_1, \dots, \ell_{k-1}}_{\ell_k}, \quad j=1,\dots,M.
\end{equation}%
Thus, applying the decomposition \eqref{eq:zetaM} to the integral in \eqref{eq:Xzetaquad1}, we can construct a quadrature form of the stochastic sojourn time model in terms of particular realisations of the deterministic switching model. The quadrature form is given by 
\begin{align}
	f(x; \Xzeta(t)) 
	&=  \int_{\R_+^M} \prod_{\ell=1}^{M} f_{\zeta_\ell\mid\mathcal{Z}_\ell}(z_\ell) f(x; \Xzetareal(t)) \d \bm z \nonumber
	\\
	&\approx \sum\limits_{{\ell_1, \dots, \ell_M}{=1}}^{L_1, \dots, L_M} \left(v_{\ell_1}v^{\ell_1}_{\ell_2}\cdots v^{\ell_1, \dots, \ell_{M-1}}_{\ell_M} \right)
	f\left(x; X^{(z_{\ell_1},z^{\ell_1}_{\ell_2}, \dots, z^{\ell_1, \dots, \ell_{M-1}}_{\ell_M}), \bm\vartheta}(t)\right) \nonumber
	\\
	&=: \sumquadL \Vquad f(x; \Xzetaquad(t)),
	\end{align}
with quadrature weights $\Vquad = \prod_{\ell_j\in|\ell|} v^{\ell_1, \dots, \ell_j-1}_{\ell_j}$.
Each density $f(x; \Xzetaquad(t))$ is the density of a randomised composite process with $M$ deterministic switches given by $\bm\tau_{\smallvert\bm\ell\smallvert} := (z_{\ell_1}, z_{\ell_1} + z^{\ell_1}_{\ell_2}, \dots, z_{\ell_1} + \dots + z^{\ell_1, \dots, \ell_{M-1}}_{\ell_M})$. For every composite process with deterministic switches, applying the quadrature approximation in \Cref{thm:DensityInterpretation} yields
	\begin{equation}
	\sum_{\smallvert\ell\smallvert = \bm 1}^{\smallvert \bm L\smallvert} \Vquad f(x; \Xzetaquad(t))
	\approx \sum_{\smallvert\ell\smallvert = \bm 1}^{\smallvert \bm L\smallvert} \Vquad \sumquad \Wquad f(x; \Xzetaquadquad(t)).
	\end{equation}
%
\par
%
The main result of this section is the stochastic switching equivalent to \Cref{thm:detswitchSDEthm}, an SDE of local-volatility type which may be defined in the framework of the probability space $(\bar\Omega, \bar\F, \bar\P)$. {The solution of this SDE exhibits marginal distributions that align with those of the  quadrature discretisation of the randomised composite process.}
\begin{thm}[Local volatility formulation of stochastic switching]\label{thm:stochswitchSDEthm}
Let $\theta_{n_j}\in\R$ and $w_{n_j}\in\R_+$, with $n_j \in \{1, \dots, N_j\}$, $N_j \in \N$, $j\in\{0, \dots, M\}$ be a collection of constants such that $\sum_{n_j=1}^{N_j} w_{n_j} = 1$.
\\
For $k\in\{1, \dots, M\}$, $L_k\in\N$, let 
\begin{equation}
z_{\ell_1}, \dots, z_{L_1}, z^{\ell_1}_{\ell_2}, \dots, z^{\ell_1}_{L_2}, \dots \dots, z^{\ell_1, \dots, \ell_{M-1}}_{\ell_M}, \dots, z^{\ell_1, \dots, \ell_{M-1}}_{L_M} \in \R^+
\end{equation}
and
\begin{equation}
v_{\ell_1}, \dots, v_{L_1}, v^{\ell_1}_{\ell_2}, \dots, v^{\ell_1}_{L_2}, \dots \dots, v^{\ell_1, \dots, \ell_{M-1}}_{\ell_M}, \dots, v^{\ell_1, \dots, \ell_{M-1}}_{L_M} \in \R^+
\end{equation}
be collections of constants such that $\sum_{\ell_j=1}^{L_j} v^{\ell_1, \dots, \ell_{j-1}}_{\ell_j} = 1$ for all $j\in\{1, \dots, M\}$.
Consider the jump-diffusion SDE
	\begin{equation}\label{eq:XlocalvolstochSDE}
	\d\Xstoch_M(t) = \widehat\mu_M\left(t, \Xstoch(t)\right) \d t + \widehat\sigma_M\left(t, \Xstoch(t)\right) \d \widehat W(t) + \widehat\eta(t) \d \widehat N(t), \quad \Xstoch(0) = x_0,
	\end{equation}
where $\widehat W(t)$ is a standard Brownian motion, $\widehat N(t)$ a Poisson process with intensity $\widehat\lambda$, and $\widehat \eta(t)$ denote the jump sizes which follow the distribution $F_{\widehat\eta}$ and arrive according to the arrival times of the process $\widehat N(t)$.
	{If the jump intensity is given by $\widehat\lambda = \lambda$, the jump size distribution has cdf $F_{\widehat\eta} = F_{\eta}$, 
	the drift coefficient and volatility coefficient are given, respectively, by
	\begin{align*}\label{eq:musolutionstoch}
	\widehat\mu_M(t,x) = 
	\frac{\sumquadL \sumquad \Vquad \Wquad \beta(t; \zquad, \thetaquad)f\left(x; \Xquadstoch(t)\right)}
	{\sumquadL \sumquad \Vquad \Wquad f\left(x; \Xquadstoch(t)\right)},
	\\
	\widehat\sigma^2_M(t, x) = 
	\frac{\sumquadL \sumquad \Vquad \Wquad \gamma^2(t; \zquad, \thetaquad)  f\left(x; \Xquadstoch(t)\right)}
	{\sumquadL\sumquad \Wquad f\left(x; \Xquadstoch(t)\right)},
	\end{align*}
	where $f\left(x; \Xquadstoch(t)\right)$ is the density of the process $\Xquadstoch(t)$, a version of the conditional process $\Xquad(t)$ in \Cref{prop:XrealSDE} with deterministic switching times determined by $\zquad = (\zeta_{\ell_1}, \dots, \zeta_{\ell_M})$, and corresponding coefficients 
	\begin{align}
	\beta(t; \zquad, \thetaquad) &= \sum_{j=0}^M b_j(s^{\zquad}_j(t), \theta_{n_j}) \1_{t \in [\zeta_{\ell_j}, \zeta_{\ell_{j+1}})}, 
	\\
	\gamma(t; \zquad, \thetaquad) &= \sum_{j=0}^M \sigma_j(s^{\zquad}_j(t), \theta_{n_j}) \1_{t \in [\zeta_{\ell_j}, \zeta_{\ell_{j+1}})}.
	\end{align}
	}
	\par
	Then, the SDE \eqref{eq:XlocalvolSDE} has a strong solution $\Xstoch_M(t)$ with probability density function
	\begin{equation}\label{eq:fXstoch}
	f_{\Xstoch_M(t)}(x) = \sumquadL \sumquad \Vquad \prod_{j=0}^M w_{n_j} f\left(x; \Xquadstoch(t)\right).
	\end{equation}
\end{thm}
The proof of this theorem is analogous to the proof of \Cref{thm:detswitchSDEthm}. 
Furthermore, we identify the characteristic function for $\Xstoch_M(t)$.
\begin{coroll}\label{cor:Xchfstochswitch}
The characteristic function of $\Xstoch_M(t)$ is given for every $u \in \R$, $t \geq 0$ and $M\in\N$ by
\begin{equation}
\varphi(u; \Xstoch_M(t)) = \e^{iux_0} \sumquadL \Vquad \prod_{j=0}^M \sum_{n_j=1}^{N_j} w_{n_j} \varphi\left(u; Y^{\theta_{n_j}}(s^{\bm z_{|\ell|}}_j(t))\right).
\end{equation}
\end{coroll}
Proof of this result can be obtained analogously to \Cref{cor:Xlvchf} with the pdf in \eqref{eq:fXstoch}.

\par
In summary, we have expanded the theory to accommodate a predetermined number of $M$ stochastic switching events. This extension suits models where the occurrence of switches is understood, but their specific timing is uncertain.
Building on this framework, the model lends itself to extensions that add further randomness to the switching behaviour. Foremost, the introduction of a random number of switches at time $t$, as this variable cannot be predicted without uncertainty in many cases.
%
\subsection{Fully stochastic switching model}
%
Let the number of regimes at time $t$ be denoted by $\mathcal{M}(t)$. The distribution of this quantity is closely linked to that of the sum of sojourn times. Observe that the event $\{\mathcal{M}(t) > m\}$ is equivalent to the event $\{\sum_{j=1}^{m+1} \zeta_j \leq t\}$, and thus we find the distribution of the number of switches from the respective complements to be
\begin{equation}
\Q_2^*\left[\mathcal{M}(t) \leq m\right] = 1 - \Q_2^*\left[ \sum_{\ell=1}^{m+1} \zeta_\ell \leq t \right].
\end{equation}
In the following, let $p_{\mathcal{M}}(t)$ denote the probability mass function of $\mathcal{M}(t)$ with support $\{0, 1, \dots, M\}$, for $M$ specified component processes. 
Depending on the choice of sojourn time distributions, the sum distribution may be analytically known%
\footnote{Notably, a finite number of exponentially distributed sojourn times with the same intensity results in an Erlang distribution. 
} 
or is computed numerically.
\par
Define $\Xfull(t)$ as the randomised composite process in which not only the switching times but also the total number of switches are stochastic. Let $\Xzeta(t; m)$ denote the randomised composite process described previously in this section, with precisely $m\in\N$ stochastic switches. 
The characteristic function of $\Xfull(t)$ can be determined with the tower property as
\begin{align}\label{eq:cffullstoch}
\varphi\left(u; \Xfull(t)\right) 
= \E\left[\exp\left(iu\Xfull(t) \right)\right] 
&= \E\left[\E\left[\left.\exp\left(iu\Xzeta(t ; m) \right)\right|\mathcal{M}(t) = m\right]\right] \nonumber \\
&= \sum_{m=0}^M \varphi\left(u; \Xzeta(t ; m)\right) p_{\mathcal{M}(t)}(m)
\end{align}
In the sum on the right-hand side, we encounter the characteristic functions of the randomised composite process with a fixed number $m$ of stochastic switches, which has been studied previously in this section. 
\begin{rem}
In this framework of fully stochastic switches, it is permissible to consider an infinite sequence of component processes, hence taking $M=\infty$. However, in practical terms, we truncate the sum in \eqref{eq:cffullstoch} at a certain $M_{max}\in\N$, chosen such that the probability $q := Q^*_2[\mathcal{M}(t) > M_{max}] < \delta$ for a specified threshold $\delta > 0$. This threshold is selected such that the resulting approximation error falls below an acceptable tolerance. 
The truncated probability mass function $p_{\mathcal{M}(t)|\mathcal{M}(t)<M_{max}}(m)$ has weights $\hat p_{\mathcal{M}(t)}(m)/(1-q)$ which ensures that 
$\varphi(0; \Xfull(t)) = 1$.
\end{rem}
\par
The local volatility model $\Xstoch_m(t)$, outlined in \Cref{thm:stochswitchSDEthm}, can be used in place of the randomised process $\Xzeta(t; m)$ which it approximates to arbitrary precision. Then, the characteristic function for the corresponding local-volatility type model with fully stochastic switching emerges.
\begin{prop}\label{prop:XfullstochChF}
Let $\Xfullstoch(t)$ denote the stochastic process corresponding to fully stochastic switching with up to $M$ switches.
Let $\Xstoch_m(t)$ be the local volatility model encoding stochastic switches and parameter randomisation as given in  \Cref{thm:stochswitchSDEthm} for a fixed number of $m\in\N$ jumps. The characteristic function of $\Xfullstoch(t)$ is given by
\begin{equation}\label{eq:chffullstochlv}
\varphi\left(u; \Xfullstoch(t)\right) = \sum_{m=0}^M \varphi\left(u; \Xstoch_m(t)\right) p_{\mathcal{M}(t)}(m).
\end{equation}
\end{prop}
The proof follows from the tower property as demonstrated in \eqref{eq:cffullstoch}.
\par
{In this section, we have extended the randomised composite process setting to account for stochastic switches, both with a deterministic and a stochastic model for the number of such switches and obtained the characteristic functions of these models.}
However, it can be observed that additional randomisation features increase the numerical complexity proportional to the number of added features and their quadrature points. Comparing the densities in \eqref{eq:quadratureinterpretationSum} and in \eqref{eq:fXstoch}, we observe that the number of summands increases from $\prod_{j=0}^M N_j$ to $\prod_{j=0}^M N_j L_j$ when the switching times are randomised. In the fully stochastic setting of \eqref{eq:chffullstochlv}, complexity is further increased to $\sum_{m=0}^{M} (\prod_{j=0}^m N_j L_j)$ terms.
\par
In the following section, we offer an alternative approach that circumvents this {curse of dimensionality} with a different construction. Instead of adding stochasticity to the switching times by randomisation, we define a Markov chain that drives the switches between the randomised component processes.

%
\section{Markov-modulation of randomised processes}\label{sec:MMRP}
%
{We propose another switching mechanism,} based on a Markov-modulated framework. Therefore, we model a continuous-time Markov process on a finite state space and consider switching times according to the state changes of the Markov process.
{The outcome is a regime-switching process reminiscent of the model proposed by \cite{Buffington2002}, enriched with randomisation features. The primary focus of this section is on deriving the characteristic function of our model, which constitutes the main result.}
In the following, we model the randomised component processes without time-dependence in the coefficients and such that their conditional versions are L\'evy processes, which preserves stationarity and greatly simplifies the proof given. 
A framework with which the approach can be extended to a broader class of processes is provided by the methodology proposed in \cite{Benth2021}.
%
%

%
%
\par
Let $R(t)$ be a Markov chain with finite state space $\mathcal{S}=\{0, 1, \dots, M\}$, generator $Q\in\R^{(M+1)\times (M+1)}$ and initial state $\bm p\in\R^{M+1}$.
For $j \in \{0, \dots, M\}$, let the randomised component process $\Y{j}(t)$ be a randomised L\'evy process of the form
\begin{equation}
\d \Y{j}(t) = b_j(\vartheta_j) \d t + \sigma_j(\vartheta_j)\d W_j(t) + \d \sum_{i=1}^{P_j(t)}\eta_i(t), \quad \Y{j}(0)=0,
\end{equation}
with  {drift and volatility coefficients, $b_j(\vartheta_j)$ and $\sigma_j(\vartheta_j)$, respectively,} that depend only on the randomiser random variable $\vartheta_j$ as defined as in \Cref{sec:randomisationSetting}. As before, $W_j(t)$ is a standard Brownian motion, {$\eta_i(t)$ are the jump sizes associated with the $i$th arrival of the Poisson process $P_j(t)$ with intensity $\lambda$.} All sources of randomness, {and the randomised component processes}, are assumed mutually independent.
\par
{The composition of these independent, randomised L\'evy processes, driven by the modulating Markov chain $R(t)$, is defined as
\begin{equation}\label{eq:Xmarkovdynamics}
\d \XMarkov := \sum_{j=0}^M \d \Y{j}(t) \1_{\{R(t) = j\}},
\end{equation}
}
with initial value $X(0; \bm \vartheta, R) = 0$.%
\footnote{Non-zero initial values can be accommodated by considering a process $x_0 + \XMarkov$.}
\par
{As before,} $\Yreal{j}(t)$ is the conditional process in which a realisation $\theta_j = \vartheta_j(\omega^*)$ is fixed. 
This is a standard L\'evy process and thus its characteristic function is explicitly available in terms of its characteristic exponent, denoted $\psireal{j}(u)$,
\begin{equation}
\varphi(u; \Yreal{j}(t)) = \exp\left(-t \psireal{j}(u)\right).
\end{equation}
By \Cref{lem:Ychf}, the characteristic function of the randomised component process is
\begin{equation}\label{eq:explicitrLPchf}
\varphi(u; Y^{\vartheta_j}_j(t))= \int_{D_j} \varphi(u; Y^{\theta_j}_j(t)) \d F_{\vartheta_j}(\theta_j) = \int_{D_j} \exp(-t \psireal{j}(u)) \d F_{\vartheta_j}(\theta_j).
\end{equation}
{The main result of this section is the characteristic function of the composite process $\XMarkov$ in the Markov-modulated L\'evy setting which} can be given in explicit form.
	\begin{thm}\label{thm:rMMLPcharfunc}
	{Let $\XMarkov$ be the Markov-modulated randomised process with dynamics given in \eqref{eq:Xmarkovdynamics}, where the underlying Markov process has initial state $\bm p$ and generator $Q$.}
	Let $A = (A_{ij})_{i,j = 1}^M$ be a diagonal matrix with entries
	{
	\begin{equation}
	A_{ij} =
		\begin{cases}
		\int_{D_j} \psireal{j}(u) \d F_{\vartheta_j}(\theta_j), & i = j, \\
		0, & i \neq j. 
		\end{cases}
	\end{equation}
	}
	Then, the characteristic function of the composite process $\XMarkov$ is given by 
	\begin{equation}
	\varphi(u; \XMarkov) = \E\left[ \e^{iu\XMarkov} \right] = \bm p \e^{(Q-A) t}  I,
	\end{equation}
	{where $I$ is the $(M+1)\times(M+1)$ identity matrix.}
	\end{thm}
The proof of this result is obtained analogous to \cite{Deelstra2017}, Lemma~A.1 with some modifications to account for the randomisation setting.
\begin{proof}[Proof of \Cref{thm:rMMLPcharfunc}]
For all times $t>0$ and states $\ell, j \in \mathcal{S}$ of the Markov chain $R(t)$, consider the case where the Markov chain originates in state $\ell$ {at $t=0$} and has arrived in state $j$ at time $t$. We denote this by
\begin{equation}
F_{\ell j}(t) := \E\left[\left. \exp(iu \XMarkov) \1_{R(t) = j} \right| R(0) = \ell \right].
\end{equation}
For small time-steps $h>0$, {we only need to consider two cases.} Either, the chain $R(t)$ is already in state $j$ where it remains, or $R(t)$ is in some state $k\neq j$ and arrives in state $j$ by time $t+h$. The possibilities of more than one jump occurring on $[t, t+h]$ can be subsumed in a term of order $o(h)$.
It thus holds that
\begin{align}
F_{\ell j}(t + h) 
&=  \E\left[\left. \e^{iu \XMarkov + iu\left(Y^{\vartheta_j}_j(t+h) - Y^{\vartheta_j}_j(t)\right)} \1_{R(t) = j} \right| R(0) = \ell \right] \nonumber 
\P[R(t+h)=j|R(t)=j] \\
&+ \smashoperator{\sum\limits_{k=0, k\neq j}^M} \E\left[\left. \e^{iu \XMarkov + iu\left(Y^{\vartheta_k}_k(t+h) - Y^{\vartheta_k}_k(t)\right)} \1_{R(t) = k} \right| R(0) = \ell \right] \nonumber \\
&\hphantom{{\sum\limits_{k=0, k\neq j}^M}} 
\cdot \P[R(t+h)=j|R(t)=k] + {o}(h). 
\end{align}
{Using the stationarity of the conditional process $Y^{\theta_j}_j$, we find that }
	\begin{multline}
	\E\left[ \e^{iu (Y^{\vartheta_j}_j(t+h) - Y^{\vartheta_j}_j(t))} \right] 
	= \int_{D_j} \E[\e^{iu (Y^{\theta_j}_j(t+h) - Y^{\theta_j}_j(t))}] \d F_{\vartheta_j}(\theta_j)
	\\
	= \int_{D_j} \E[\e^{iu Y^{\theta_j}_j(h)}] \d F_{\vartheta_j}(\theta_j) 
	= \E\left[ \e^{iu Y^{\vartheta_j}_j(h)}\right].
	\end{multline}
{For every state $k \in \mathcal{S}$, we may thus rewrite each of the terms in the sum above as}
\begin{multline}
\E\left[\left. \e^{iu \XMarkov + iu(Y^{\vartheta_k}_k(t+h) - Y^{\vartheta_k}_k(t))}  \1_{R(t) = k}\right| R(0)=\ell\right] 
\\
= \E\left[\left. \e^{iu \XMarkov} \1_{R(t) = k}\right| R(0)=\ell\right] \E\left[\e^{iu(Y^{\vartheta_k}_k(t+h) - Y^{\vartheta_k}_k(t))} \right]  
\\
= F_{\ell k}(t) \E\left[ \e^{iuY^{\vartheta_k}_k(h)} \right] 
= F_{\ell k}(t) \int_{D_\ell} \e^{-h \psireal{k}(u)} \d F_{\vartheta_k}(\theta_k).
\end{multline}
The transition probabilities of the Markov process $R$ are explicitly available in terms of its generator $Q$,
\begin{align}
\P[R(t+h)=j|R(t)=j] &= 1 + Q_{jj}h + o(h), \\
\P[R(t+h)=j|R(t)=k] &= Q_{kj}h + o(h).
\end{align}
Consequently, $F_{\ell j}(t+h)$ is given by
\begin{align}
F_{\ell j}(t + h) 
&= F_{\ell j}(t) \int_{D_j} \e^{- \psireal{j}(u)h} \d F_{\vartheta_j}(\theta_j) (1 + Q_{jj}h + o(h))  \nonumber\\
&\quad + \sum\limits_{k=0, k\neq j}^M F_{\ell k}(t) \int_{D_k} \e^{ - \psireal{k}(u)h} \d F_{\vartheta_k}(\theta_k) (Q_{kj}h + o(h)).
\end{align}
{With the first order expansion of the exponential function} 
$\exp(zh) = 1 + zh + o(h)$, 
we can rewrite the above equation to
\begin{multline}\label{eq:5.14}
F_{\ell j}(t + h) = F_{\ell j}(t) (1 - \Psi_j(u)h + o(h)) (1 + Q_{jj}h + o(h)) \\
+ \sum\limits_{k=0, k\neq j}^M F_{\ell k}(t) (1 - \Psi_k(u)h + o(h)) (Q_{kj}h + o(h)),
\end{multline}
{where we denote the integrated characteristic exponent by
\begin{equation}
\Psi_j(u) := \int_{D_j} \psireal{j}(u) \d F_{\vartheta_j}(\theta_j).
\end{equation}
}
By expanding \eqref{eq:5.14} and collecting all the terms of order $o(h)$, one can show that the expression is equivalent to
\begin{equation}
F_{\ell j}(t+h) = F_{\ell j}(t) - F_{\ell j}(t)\Psi_j(u)h + \sum\limits_{k=0}^M F_{\ell k}(t) Q_{kj} h + o(h).
\end{equation}
Define the {$(M+1)\times (M+1)$ matrix $F(t) := (F_{\ell j}(t))_{\ell, j = 0}^M$}. Then {the above equation may be rearranged to}
\begin{equation}
\frac1h (F(t+h) - F(t) - o(h)) = \left(-F_{\ell j}(t)\Psi_j(u) + \sum\limits_{k=0}^M F_{\ell k}(t) Q_{kj} \right)_{\ell,j=0}^M.\label{eq:5.21}
\end{equation}
It can be verified that the right-hand side of \eqref{eq:5.21} coincides with the entries of the matrix $F(t)Q - F(t) A$, where $Q$ is the generator of the Markov process $R(t)$ and $A$ is given by 
\begin{equation}
A_{jk} :=
\begin{cases}
0,& j \neq k, \\
\Psi_j(u), & j = k.
\end{cases}
\end{equation}
Finally, taking the limit of $h\to0$ we find that
\begin{equation}
\frac{\d}{\d t} F(t) = F(t) Q - F(t) A,
\end{equation}
which has a solution $F(t) = \e^{(Q-A)t}$.
{Gathering all} state combinations $j, \ell \in \mathcal{S}$ for $F_{\ell j}(t)$ yields the result,
\begin{equation}
\E\left[ e^{iu\XMarkov} \right] = \sum\limits_{\ell=1}^M \sum\limits_{j=1}^M F_{\ell j}(t) \P[R(0)=\ell] = \bm p e^{(Q-A)t} I.
\end{equation}
\end{proof}
%
%
%
\section{{Numerical Study with a Financial Application}}\label{sec:numerics}
This section provides a brief overview of the local volatility models $\bar X(t), \widehat X(t), \widetilde X(t)$ associated with composite randomised processes proposed in this article. 
We examine a log-price asset model enriched with randomisation and switching. First, we present path realisations of the deterministic and stochastic switching cases, giving a more detailed view of their behaviours. 
{Then, we consider an application to the financial problem of option pricing, in which} we compare the shapes of the implied volatility curves for a European option under the different switching assumptions and additionally contrast them with a randomised setting without switches between randomiser regimes.
\par
Consider an asset with price given by $S(t) = S(0)\e^{\mathbb{X}(t)}$, where $\mathbb{X}\in\{\bar X, \widehat X, \widetilde X\}$ represents the local volatility models associated with the composite randomised processes with deterministic, stochastic and fully stochastic switching, respectively.
For every $j\in\{0, \dots, M\}$, the component processes of \Cref{dfn:componentprocess} can be understood as randomised versions of the model of \cite{Merton1976} {if we choose the drift functions to be}
\begin{equation}
b_j(\vartheta_j) = r - \frac{\sigma_j^2(\vartheta_j)}{2} - \lambda k_j,
\end{equation}
where $k_j:=\E[\e^{\eta_j}-1]$ and $r\geq0$ represents the risk-free interest rate\footnote{Throughout this section, we select $r=0.05$ and note that this parameter is of no concern for the model features shown.}. 
In the following, we {also eliminate} the jump component $\d  \sum_{i=1}^{P_j(t)} \eta_j(i)$, so that each component process $\Y{j}$ can be connected to a randomised Black--Scholes model. The parameter randomisers $\vartheta_j$ are modelled with a normal distribution, and the volatility functions are given by $\sigma_j(\vartheta_j) = \vartheta_j$.\footnote{The boundedness criterion for $\sigma_j$ is fulfilled in all practical applications by some maximal sample $\vartheta_j(\omega^*)<\infty$.} %
The resulting randomised component processes considered in this section are
\begin{equation}
\d \Y{j}(t) = \left(r - \frac{\vartheta_j^2}{2}\right)\d t + \vartheta_j \d W_j(t).
\end{equation}
\par
{The regimes associated with normally distributed randomisers $\vartheta_j$ are characterised by the two parameters of the distributions, $\vartheta_j \sim \mathcal{N}(\nu_j, \xi_j^2)$. The mean $\nu_j$ prescribes the average level of volatility within the regime, and the standard deviation $\xi_j$ can be interpreted as the certainty about this volatility level, or as a measure of the expected fluctuation around the volatility level.}
\par
We consider a scenario where the seasonality of the asset fluctuates between two regimes representing a `calm' and an `excited' period. The calm regime is represented by randomisers with low mean $\nu_j$ and low variance $\xi_j^2$, whereas both these parameters are increased during the excited period, corresponding to elevated volatility and more fluctuation about this level. We establish a repeating pattern by imbuing all even-numbered randomisers with the `calm' distribution, $\vartheta_0, \vartheta_2, \dots \sim \mathcal{N}(\nu_0, \xi_0^2)$, and all odd-numbered randomisers with the `excited' distribution $\vartheta_1, \vartheta_3, \dots \sim \mathcal{N}(\nu_1, \xi_1^2)$. Whenever quadrature pairs are computed from randomisers or sojourn times, we compute with $N_j = L_j = 7$ quadrature points for all $j\geq 0$.
\begin{figure}[t]
\centering
\includegraphics[width=.5\textwidth]{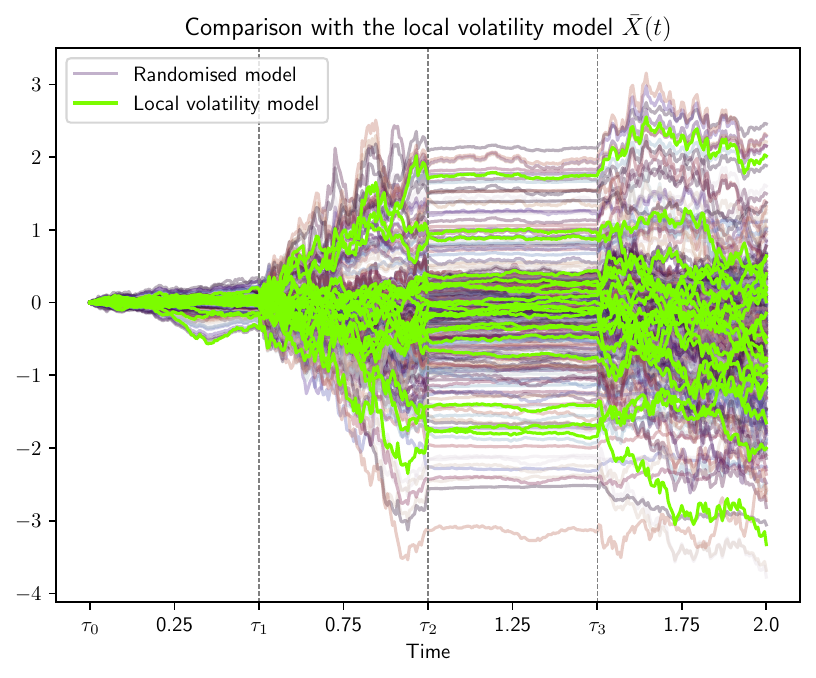}\includegraphics[width=.5\textwidth]{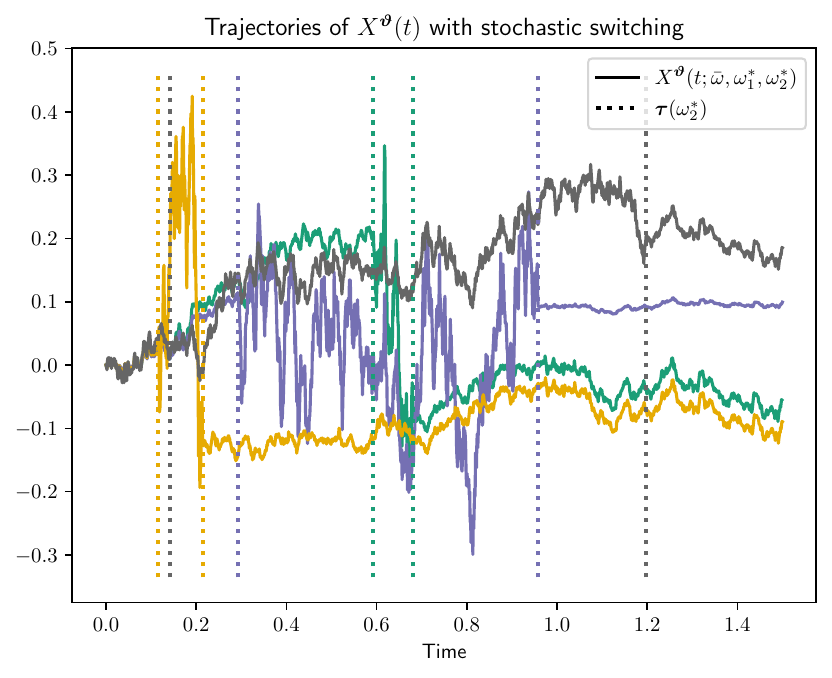}
\caption[Path simulation: Deterministic switching local volatility, stochastic switching randomised]
{Left: Sample paths of the randomised model with deterministic switching are contrasted with its associated local volatility model $\Xlv(t)$. Right: Sample paths of the stochastic switching model with two switches. Every trajectory uses the same underlying Brownian motion, all differences stem from the random samples of parameter randomisers and sojourn times.}
\label{fig2}
\end{figure}
\par
In \Cref{fig1}, realisations of the randomised composite process $\X(t)$ with deterministic switching, as defined in \Cref{def:Xprocess}, are given. 
The realisation of the underlying Brownian motion is always the same, only the realisations of the randomisers $\vartheta_j(\omega^*)$ vary. 
The trajectories exhibit distinct patterns of low and high volatility as they pass through the different regimes.
In \Cref{fig2} on the left side, trajectories of the associated local volatility model $\Xlv(t)$ as defined in \Cref{thm:detswitchSDEthm} are given. They are contrasted with realisations of the randomised model $\X(t)$, this time for both different randomiser and Brownian motion samples.
Notably, within each regime, every trajectory of the local volatility model exhibits the same volatility and drift. We no longer observe the differences between trajectories of the randomised model as in \Cref{fig1}. Otherwise, the process structure is maintained, as the construction ensures {consistency in the marginal distributions.}
\par 
On the right side of \Cref{fig2}, trajectories of $\Xzeta(t)$ with stochastic switching as described in \Cref{def:Xprocessstoch} are given. The sojourn times $\zeta_j \sim \mathcal{E}\mathrm{xp}(2)$, $j=0,1,2$, are exponentially distributed with a mean of $1/2$ and conditioned on $M=2$ switches on the time interval $[0, 1.5]$. As in \Cref{fig1}, the underlying Brownian motion sample is the same for each path, whereas different samples of the parameter randomisers $\vartheta_j$ and the sojourn times $\zeta_j$ are taken. The observed low and high volatility patterns persist in the regimes, with the distinction that switching times are now specific to each sample path. Even though each path in the figure is driven by the same Brownian motion, the pathwise behaviour varies widely. 
\begin{figure}[t]
\centering
\includegraphics[width=\textwidth]{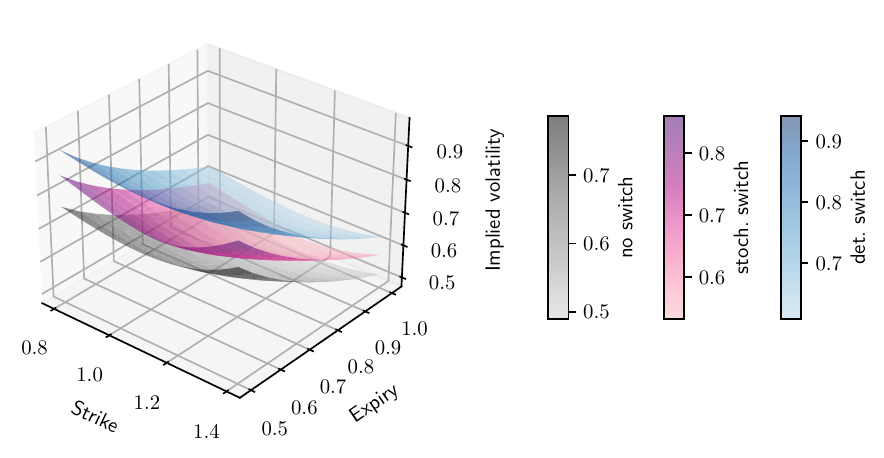}
\caption[IV surface: Single switch against expiry]
{Implied volatility surface as the expiry $T$ of the option decreases. 
The underlying asset is modelled with one switch at time $T/2$ in the deterministic switch model $\Xlv(t)$, and one exponentially distributed switch such that $\E[\zeta_1] = T/2$ in the stochastic switching model $\Xstoch(t)$, both times the composite process switches from a `calm' into an `excited' randomiser regime. We also consider the randomised model with no switch, where the randomiser is that of the excited regime throughout.}
\label{fig3}
\end{figure}
\par
We continue with a study of implied volatility (IV) surfaces, in which we compare the surfaces obtained when the underlying asset is modelled with the composite process under different switching rules. In all cases, we consider a European call option and compute the implied volatility against a normed strike, i.e., a strike of $1$ corresponds to the at-the-money (ATM) option. We model the asset with the local volatility models corresponding to deterministic switching, $\Xlv(t)$, stochastic switching, $\Xstoch(t)$, and fully stochastic switching $\Xfullstoch(t)$. Furthermore, we compare with the randomised model in which no switch occurs. For all these models, the characteristic functions have been obtained in \Cref{cor:Xlvchf}, \Cref{cor:Xchfstochswitch}, \Cref{prop:XfullstochChF}, and \Cref{lem:Ychfquad}, respectively. Based on the characteristic functions, we compute all option prices needed for implied volatility computations with the COS-method of \cite{fang2009cos}.
\par
In \Cref{fig3} we consider the classic volatility surface spanned by a strike range $K\in[0.8, 1.4]$ and an expiry $T\in[.5, 1]$. Three models for the underlying asset's (log)price are being compared. There is the deterministic switching model $\Xlv(t)$ with one switch at time $\tau_1 = T/2$ from a `calm' randomiser $\vartheta_0 \sim \mathcal{N}(0.15, 0.1^2)$ to an excited randomiser $\vartheta_1 \sim\mathcal{N}(0.3, 1)$. The same randomiser sequence is used in the stochastic switching model $\Xstoch(t)$ with an exponentially distributed sojourn time $\zeta(1)\sim\mathcal{E}\mathrm{xp}(2T)$ such that $\E[\zeta_1] = T/2$. Finally, we also consider a randomised model without a switch, in which the randomiser is that of the excited regime, $\vartheta_1$, throughout.
The resulting IV surfaces are ordered, with deterministic switching surpassing stochastic switching, which, in turn, exceeds the case with no switching at every expiry. The sensitivity to expiry, i.e., the slope of the IV surfaces, appears to be consistent across the three models.  We observe a mild amount of smile and skew in the implied volatilities obtained from all three randomised models, which increases as the expiry is reduced.
\begin{figure}[t]
\centering
\includegraphics[width=\textwidth]{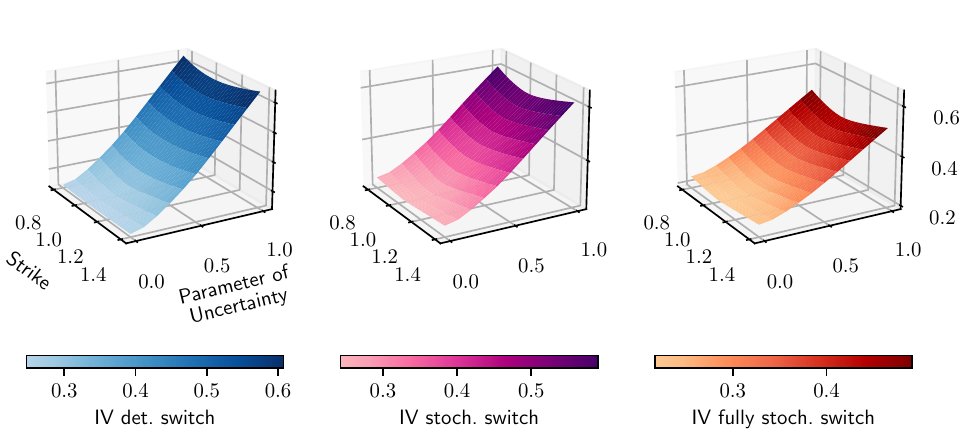}
\caption[IV surfaces in randomised market states]
{We consider the two-regime example of a `calm' and an `excited' market state, with randomiser distribution \protect{$\mathcal{N}(\nu_0, \xi_0^2)$} during the `calm' regime and randomiser distribution \protect{$\mathcal{N}(\nu_1, \xi_1^2)$} during the `excited' regime. The implied volatility surfaces are obtained for a range of values for the `excited' randomiser's standard deviation \protect{$\xi_1$}. Considered are the deterministic and stochastic switching models with one switch each, as well as the fully stochastic switching model with a random number of switches.}
\label{fig:fig4.4}
\end{figure}
Another type of implied volatility surface is studied in \Cref{fig:fig4.4}, where the expiry remains fixed at $T=1$. Instead, the IV surface is spanned between strikes $K\in[0.8, 1.4]$ and a range for the standard deviation $\xi_1\in[0, 1]$ of the randomiser $\vartheta_1$, which is associated with the `excited' regime. With the resulting surfaces, it is possible to observe the sensitivity of model IV to the variance of one of its randomisers. The randomiser $\vartheta_0$ of the calm regime remains unchanged and we again consider the models $\Xlv(t)$ and $\Xstoch(t)$ with one deterministic, respectively stochastic, switch. Additionally, we consider the fully-stochastic switching model $\Xfullstoch(t)$ in which multiple switches can occur, alternately between randomisers associated with the calm and the excited regimes. All sojourn times are i.i.d. $\mathcal{E}\mathrm{xp}(2)$ distributed. In the computations, the $\Xfullstoch(t)$ model is limited to a maximum of $4$ switches, as the probability that $5$ or more switches occur is less than $0.05$ with negligible impact on the resulting implied volatility surfaces. 
Again, all three IV surfaces exhibit a mild `smile' shape that is more pronounced as the randomiser standard deviation $\xi_1$ increases, and again the IV surfaces are ordered. This time, at each fixed value of $\xi_1$, the implied volatility of the deterministic switching model is slightly larger than that of the stochastic switching model, which is larger than the implied volatility of the fully stochastic switching model. The sensitivities to the studied parameter, expressed through the IV surface slopes, appear to also be ranked in the same order. An explanation for the reduced IV in the fully stochastic model is found in the influence of the no-switch case arising when the stochastic number of switches at expiry $\mathcal{M}(1)$ is zero.
Overall, both experiments show that randomisation and switching have a significant impact on the model's implied volatility. 
%
\section{Conclusion}\label{sec:conclusio}
{We construct local volatility models that emulate stochastic processes wherein drift and diffusion coefficients are contingent on random variables, referred to as randomisers. Additionally, the dynamics and governing random variables change between time regimes determined either deterministically or stochastically.}
The resulting local volatility models do not require any non-standard definitions of a stochastic process which allows for a classical treatment. Their form relates to specific weighted sums of the mimicked processes, where the randomisers are replaced with particular realisations, i.e., deterministic values.
These values and weight pairs are obtained by applying the Gauss quadrature technique to integrals against the density functions of the randomisers. Their computation with the Golub-Welsch algorithm only requires knowledge about the moments of the randomisers.
We compute characteristic functions for the local volatility models corresponding to deterministic switching, stochastic switching with a known number of switches, and fully stochastic switching where additionally the number of switches is modelled stochastically.
We also formulate a Markov-modulated model in which a Markov process governs the switching between randomised processes.
Numerical experiments show that randomisation and switching have a sizeable impact on implied volatilities, indicating a large impact of modelling externalities that would cause regimes or randomisation. 
\section*{Acknowledgments}
This research is part of the ABC--EU--XVA project and has received funding from the European Union's Horizon 2020 research and innovation programme under the Marie Sk\l{}odowska--Curie grant agreement No.\ 813261.
\bibliography{Lit}
\input{Appendix}

\end{document}

%% file: Appendix.tex
%
\appendix
\section{Proof of \Cref{prop:XrealSDE}}\label{appx:XrealSDEproof}
\begin{proof}
By \Cref{dfn:componentprocess} and \Cref{def:Xprocess}, it holds that
	\begin{equation}
	\Xreal{}(t) = x_0 + \sum_{j=0}^M \left( \int_0^{s_j(t)} b_j(u, \theta_j)\d u +  \int_0^{s_j(t)} \sigma_j(u, \theta_j) \d W_j(u) +\sum_{i_j=1}^{P_j(s_j(t))} \eta_j(i_j)\right).
	\end{equation}
{With a change of variables according to the definition of the time shift $s_j(t)$ in \eqref{def:sj} shows that the drift part equals}
	\begin{multline}
	 \sum_{j=0}^M \int_0^{s_j(t)} b_j(u, \theta_j)\d u 
	 = \sum_{j=0}^M \int_{\tau_j \wedge t}^{\tau_{j+1}\wedge t} b_j(s_j(u), \theta_j)\d u
	 \\
	 = \int_0^t \sum_{j=0}^M b_j(s_j(u), \theta_j)  \1_{u \in [\tau_j, \tau_{j+1})} \d u 
	 = \int_0^t \beta(u; \bm\theta) \d u.
	\end{multline}
Similarly, it holds for the diffusion part that
	\begin{align}
	\sum_{j=0}^M \int_0^{s_j(t)} \sigma_j(u, \theta_j) \d W_j(u) 
	&= \sum_{j=0}^M \int_{\tau_j \wedge t}^{\tau_{j+1}\wedge t} \sigma_j(s_j(u), \theta_j) \d W_j(s_j(u)) 
	\nonumber
	\\
	&= \sum_{j=0}^M \int_{\tau_j \wedge t}^{\tau_{j+1}\wedge t} \sigma_j(s_j(u), \theta_j) \d \left[\sum_{k=0}^M W_k(s_k(u)) \right],
	\end{align}
where we used that $\d \sum_{k=0}^M W_k(s_k(u)) = \d W_j(s_j(u))$ for $u\in(\tau_j, \tau_{j+1})$.
By setting 
	\begin{equation}\label{eq:Bdefproof}
	\widetilde W(t) := \sum_{k=0}^M W_k(s_k(t)),
	\end{equation}
the diffusion component becomes
	\begin{equation}
	\sum_{j=0}^M \int_0^{s_j(t)} \sigma_j(u, \theta_j) \d W_j(u)
	= \int_0^t \sum_{j=0}^M \sigma_j(s_j(u),  \theta_j) \1_{u \in [\tau_j, \tau_{j+1})} \d \widetilde W(u) = \int_0^t \gamma(u; \bm \theta)\d \widetilde W(u).
	\end{equation}
\par
Next, we show that $\widetilde W(t)$ is a standard Brownian motion. It is immediate from \eqref{eq:Bdefproof} that
{$\widetilde W(t)$ is a Gaussian process with almost surely $\widetilde W(0) = 0$ and continuous paths.}
Let $0 \leq u < t$ and denote by $L := \sup\{j\colon \tau_j \leq u\}$ and $K := \sup\{j\colon \tau_j \leq t\}$ the indices of the final switching times before $u$ and $t$, respectively. Then, by the independence of $W_j$ and $W_\ell$ for $j\neq \ell$, it holds that
	\begin{multline}
	\Cov[\widetilde W(t), \widetilde W(u)] 
	= \Cov\left[ \sum_{j=0}^{K} W_j(s_j(t)), \sum_{\ell=0}^{L} W_\ell(s_\ell(u)) \right] 
	\\
	= \sum_{j=0}^{L-1} \Var\left[ W_j(\tau_{j+1}-\tau_j)\right] + \Cov[W_{L}(s_{L}(t)), W_{L}(s_{L}(u))]
	\\
	= \sum_{j=0}^{L-1} (\tau_{j+1} - \tau_j) + (u - \tau_{L}) = u,
	\end{multline}
therefore, $\widetilde W(t)$ is a standard Brownian motion. 
\par
It remains to show that the jump component has the desired shape,
	\begin{equation}
	\sum_{j=0}^M \sum_{i_j=1}^{P_j(s_j(t))} \eta_j(i_j) = \sum_{k=1}^{\widetilde P(t)} \pi(k).
	\end{equation}
As all the jump sizes $\eta_j$ have the same cdf $F_\eta$, which is also the distribution of $\pi$, the result follows by showing that the concatenation of Poisson processes $P_j(s_j(t))$ is again a Poisson process,
	\begin{equation}
	\widetilde P(t) := \sum_{j=0}^M P_j(s_j(t)).
	\end{equation}
We note that $\widetilde P(0) = \sum_{j=0}^M P_j(0) = 0$ almost surely, and let $L$ and $K$ be the last switching time index before times $u$ and $t$, respectively, with $u<t$, as before.
Observe that it holds
	\begin{equation}
	\widetilde P(t) - \widetilde P(u) 
	= \sum_{j=L+1}^M P_j(s_j(t)) + (P_L(s_L(t)) - P_L(s_L(u))),
	\end{equation}
which, by the independence of the processes $P_j$, is a Poisson distributed random variable with parameter $\lambda \left(\sum_{j=L+1}^M s_j(t) + s_L(t) - s_L(u)\right)$.
Identifying the telescopic sum, we find this parameter to be
	\begin{equation}
	\lambda \left(\sum_{j=L+1}^M s_j(t) + s_L(t) - s_L(u) \right) = \lambda (t - u).
	\end{equation}
It remains to show independence between increments. Let $v<u<t$ be an arbitrary partition and let $J := \sup\{j\colon \tau_j \leq v\}$. Then, by the previous argument, it holds that
	\begin{equation}
	\widetilde P(t) - \widetilde P(u) = \sum_{j=L+1}^M P_j(s_j(t)) + (P_L(s_L(t)) - P_L(s_L(u))),
	\end{equation}
and
	\begin{equation}
	\widetilde P(u) - \widetilde P(v) = \sum_{j=J+1}^L P_j(s_j(u)) + (P_J(s_J(u)) - P_J(s_J(v))).
	\end{equation}
Independence between $\widetilde P(t)-\widetilde P(u)$ and $\widetilde P(u)-\widetilde P(v)$ follows immediately by the independence between the Poisson processes $P_j$ and $P_k$ for $j\neq k$, and by the independence of the increment $P_L(s_L(t)) - P_L(s_L(u))$ from $P_L(s_L(u))$. 
\end{proof}
\section{Probability density functions of the randomised component and composite process}\label{appx:RandomisedPDF}
{Analogous to \Cref{lem:Ychf} and \Cref{lem:Ychfquad}, we find the pdf of the randomised component processes $\Y{j}(t)$ and its quadrature approximation.}
	\begin{coroll}\label{cor:Ydensity}
	{Let the conditions of \Cref{lem:Ychf} and \Cref{lem:Ychfquad} hold and let the pdf of $\Yreal{j}(t)$ be denoted by $f(y; \Yreal{j}(t))$.}
	Then, the probability density function of the randomised component process is given by
		\begin{equation}\label{eq:Ydensiquad}
		f(y; \Y{j}(t)) 
		= \int_{D_j} f(y; Y^{\theta_j}_j(t)) \d F_{\vartheta_j}(\theta_j)
		= \sum_{n_j=1}^{N_j} w_{n_j} f(y; \Yquad{n}{j}(t)) + \hat\varepsilon_{N_j}(t, y).
		\end{equation}
	The Gauss-quadrature pairs  $(w_{n_j}, \theta_{n_j})_{n_j=1}^{N_j}$ are the same as in \Cref{lem:Ychf}, and $\hat\varepsilon_{N_j}(t, y)$ is the quadrature approximation error bounded by 
		\begin{equation}\label{eq:Ypdfquaderror}
		\hat\varepsilon_{N_j}(t, y) \leq \sup\limits_{\xi \in D_j} \frac{1}{(2N)!} \left.\frac{\partial^{2N}}{\partial \theta_j^{2N}} f(y; Y^{\theta_j}_j) \right|_{\theta_j = \xi} .
		\end{equation}
	\end{coroll}
This result can be proved by the same means as \Cref{lem:Ychf}, with the characteristic function replaced by the probability density function. In practice, however, the density of $\Yreal{j}(t)$ is not always explicitly known, e.g., in the presence of jump terms. In this section, we base the proof on the Fourier transformation relationship between pdf and chf, which allows a connection to the approximation bound \eqref{eq:Ychfquaderror} of the characteristic function. This approach is also used in the proof of one of the main results, \Cref{thm:DensityInterpretation}.

	\begin{proof}
	Let $\mathfrak{F}$ denote the (probabilist's) Fourier transformation, and $\mathfrak{F}^{-1}$ its inverse. 
	Utilizing the linearity of the (inverse) Fourier transform, we obtain the density of the randomised component process as
	\begin{align*}
	f(y; Y{j}(t)) 
	= \mathfrak{F}^{-1}\varphi(u; \Y{j}(t)) 
	&= \mathfrak{F}^{-1}\int_{D_j} \varphi(u; Y^{\theta_j}_j(t)) \d F_{\vartheta_j}(\theta_j)
	\\
	&= \int_{D_j} \mathfrak{F}^{-1}\varphi(u; Y^{\theta_j}_j(t)) \d F_{\vartheta_j}(\theta_j).
	\end{align*}
	\end{proof}
%
{We proceed with the proof of \Cref{thm:DensityInterpretation}.}
\begin{proof}[Proof of \Cref{thm:DensityInterpretation}, i]
Let $\X(t)$ be the randomised composite process with conditional version $\Xreal(t)$ given in \Cref{def:Xprocess}.
Analogous to the proof of \Cref{lem:Ychf}, we find the characteristic function $\varphi(u;\X(t))$ to be
	\begin{equation}\label{eq:chfX1}
	\varphi(u;\X(t)) = \E\left[ \e^{iu\X(t)} \right] 
	= \E\left[ \E\left[\left.\e^{iu\Xreal(t)} \right| \theta = \vartheta(\omega^*)\right]\right]
	= \E\left[ \varphi(u;\Xreal(t)) \right].
	\end{equation}
By the independence of the component processes $\Yreal{j}$, $\Yreal{k}$ for $j\neq k$, the characteristic function of $\Xreal(t)$ can be factored into
	\begin{equation}
	\varphi(u;\Xreal(t)) = \varphi\left(u;x_0 + \sum_{j=0}^M \Yreal{j}(s_j(t))\right)
	= \e^{iux_0} \prod_{j=0}^M \varphi\left(u;\Yreal{j}(s_j(t))\right).
	\end{equation}
Inserting this back into \eqref{eq:chfX1} and utilizing the independence of $\vartheta_j$, $\vartheta_k$ for $j\neq k$, we obtain
	\begin{equation}
	\varphi(u;\X(t)) 
	= \e^{iux_0} \prod_{j=0}^M \E\left[ \varphi\left(u; Y^{\theta_j}_j(s_j(t))\right) \right]
	= \e^{iux_0}\prod_{j=0}^M \varphi(u; \Y{j}(s_j(t))). \label{eq:chfXproduct}
	\end{equation}
The pdf of $\X(t)$ in \eqref{eq:Xdensi} follows immediately from the convolution theorem of Fourier transformation, 
	\begin{align}
	f_{\X(t)}(x) 
	&= \mathfrak{F}^{-1} \varphi_{\X(t)}(x) 
	= \mathfrak{F}^{-1} \left(\e^{iux_0}\prod_{j=0}^M \varphi_{\Y{j}(s_j(t))}\right)(x)
	\\
	&= \left(\mathfrak{F}^{-1} \e^{iux_0} \ast \mathfrak{F}^{-1} \varphi_{\Y{0}(s_0(t))}(x) \ast \cdots \ast \mathfrak{F}^{-1} \varphi_{\Y{M}(s_M(t))}\right)(x) \label{eq:convolutiontheoremstep}
	\\
	&= \left(\delta_{x_0} \ast  f_{\Y{0}(s_0(t))} \ast \cdots \ast  f_{\Y{M}(s_M(t))}\right)(x),
	\end{align}
where we have used that $\mathfrak{F}^{-1} \e^{iux_0} (x)= \delta_{x_0}(x)$, 
with $\delta_{x_0}(x)$ the translated Dirac delta function.
\end{proof}
The proof of the second part of \Cref{thm:DensityInterpretation} relies on the result that the density of a sum of independent random variables is the convolution of the densities, and on the Fourier convolution theorem which connects the product shape of the characteristic function $\varphi_{\X(t)}(u)$ with a convolution of the densities. 
An additional observation is necessary before we proceed to prove the remaining part of the theorem.
	\begin{prop}\label{prop:quaderrorintegral}
	For every $j\in\{0,\dots,M\}$ and $t\geq 0$, the integral over the space domain of the {randomised component density} quadrature error vanishes, 
	$
	\int_\R \hat\varepsilon_{N_j}(t, y) \d y = 0.
	$
	\end{prop}
\begin{proof}
By the definition of a density, it holds that $\int_\R f\left(y; \Y{j}(t)\right) \d y  = 1$.
Therefore, inserting the discretisation \eqref{eq:Ydensiquad} yields
	\begin{align}
	1 
	&= \int_\R f\left(y; \Y{j}(t)\right) \d y = \int_\R \left( \sum_{n_j=1}^{N_j} w_{n_j} f\left(y; \Yreal{j}(t)\right) + \hat\varepsilon_{N_j}(t, y) \right) \d y
	\\
	&= \sum_{n_j=1}^{N_j} w_{n_j} \int_\R f\left(y; \Yreal{j}(t)\right) \d y + \int_\R \hat\varepsilon_{N_j}(t, y) \d y
	= \sum_{n_j=1}^{N_j} w_{n_j} 1 + \int_\R \hat\varepsilon_{N_j}(t, y) \d y
	\\
	&= 1 + \int_\R \hat\varepsilon_{N_j}(t, y) \d y,
	\end{align}
and thus $ \int_\R \hat\varepsilon_{N_j}(t, y) \d y = 0$.
\end{proof}
\begin{proof}[Proof of \Cref{thm:DensityInterpretation}, ii.]
The convolution representation \eqref{eq:Xdensi} in \Cref{thm:DensityInterpretation}, i.\ can be expressed in integral form {by successive application of the associative law} as
\begin{equation}\label{eq:compdensquadintegralform}
f\left(x; \X(t)\right)
= \int_{\R^M}\int_\R \delta_{x_0}(y) 
		\left( \prod_{j=0}^{M-1} 
		f\left(z_j; \Y{j}(s_j(t))\right)
		\right)
		f\left(z_M;\Y{M}(s_M(t)) \right) 
		\d y \d \bm z,
\end{equation}
where we set 
$z_M := x - y - z_0 - \dots - z_{M-1}$
and gather the integration dummy variables (excluding $y$) in
$\bm z := (z_0, \dots, z_{M-1})$.
In the following, we abbreviate the quadrature sum by $\Gamma_j(z_j) := \sum_{n_j=1}^{N_j} w_{n_j} f\left(z_j; \Yquad{n}{j}(s_j(t))\right)$.
Then, substitution of $f\left(z_j; \Y{j}(s_j(t))\right)$
{in \eqref{eq:compdensquadintegralform} with its quadrature form, given in \eqref{eq:Ydensiquad},} yields
\begin{equation}\label{eq:B10}
f\left(x; \X(t)\right) 
		= \int_{\R^M}\int_\R \delta_{x_0}(y) 
		\prod_{j=0}^M \left( \Gamma_j(z_j) + \hat\varepsilon_{N_j}(t, z_j)\right)
		 \d y \d \bm z.
\end{equation}
We take a closer look at the product term $\prod_{j=0}^M \left( \Gamma_j(z_j) + \hat\varepsilon_{N_j}(t, z_j)\right)$. Factoring out the entire product yields three distinct types of summands.
First, the summand $\prod_{j=0}^M \Gamma_j(z_j)$, which contains no error terms $\hat\varepsilon_{N_j}(t, z_j)$.
Secondly, there are the $M+1$ summands that contain exactly one error term each, given by
\begin{equation}
G_j(\bm z) := \hat\varepsilon_{N_j}(t, z_j) \prod_{\substack{k=0 \\ k\neq j}}^M \Gamma_k(z_k).
\end{equation}
Finally, there are all the remaining mixed summands with two or more error terms, which we denote by $H(\bm z)$.
It thus holds
\begin{equation}
\prod_{j=0}^M \left( \Gamma_j(z_j) + \hat\varepsilon_{N_j}(t, z_j)\right) = \prod_{j=0}^M \Gamma_j(z_j) + \sum_{j=0}^M G_j(\bm z) + H(\bm z).
\end{equation}
Substituting this into \eqref{eq:B10} yields
\begin{multline}\label{eq:multibinomialsplit}
f\left(x; \X(t)\right) = \int_\R\int_{\R^M} \delta_{x_0}(y) \prod_{j=0}^M \Gamma_j(z_j) \d y \d \bm z
\\
+ \int_\R\int_{\R^M} \delta_{x_0}(y) \sum_{j=0}^M G_j(z_j) \d y \d \bm z
\, + \int_\R\int_{\R^M} \delta_{x_0}(y) H(\bm z) \d y \d \bm z.
\end{multline}
We recognize the first summand as the convolution of quadratures,
\begin{equation}
\int_\R\int_{\R^M} \delta_{x_0}(y) \prod_{j=0}^M \Gamma_j(z_j) \d y \d \bm z 
= \left(\delta_{x_0} \ast \Gamma_0 \ast \cdots \ast \Gamma_M\right) (x).
\end{equation}
The second set of integrals in \eqref{eq:multibinomialsplit} can be computed as
\begin{align}\label{eq:B16}
\int_\R\int_{\R^M} \delta_{x_0}(y) \sum_{j=0}^M G_j(z_j) \d y \d \bm z
&=\sum_{j=0}^M \int_\R\int_{\R^M} \delta_{x_0}(y)  G_j(z_j) \d y \d \bm z
\nonumber
\\
&=\sum_{j=0}^M \int_\R \delta_{x_0}(y) \int_{\R^M} \hat\varepsilon_{N_j}(t, z_j) \prod_{\substack{k=0\\k\neq j}}^M \Gamma_k(z_k) \d y \d \bm z.
\end{align}
{By bounding the summand where $j=M$, $\hat\varepsilon_{N_j}(t, z_M)$, with its supremum over $z$, it becomes possible to factor it out of the integrals,} so that $M+1$ integrators $(y, z_0, \dots, z_{M-1})$ remain for $M+1$ independent integrands, allowing those to be solved consecutively. Each such integral equals $1$, respectively, by the definition of the Dirac delta and since $\Gamma_k(z_k)$ is a weighted sum of probability density functions:
\begin{multline}
\int_\R \delta_{x_0}(y) \int_{\R^M} \hat\varepsilon_{N_M}(t, z_M) \prod_{\substack{k=0}}^{M-1} \Gamma_k(z_k) \d y \d \bm z
\\
\leq \left( \int_\R \delta_{x_0}(y)\d y {\int_{\R^M}} \prod_{j=0}^{M-1}\Gamma_k(z_k) \d \bm z\, \right) \sup\limits_{z\in\R} \hat\varepsilon_{N_M}(t, z)
\leq \sup\limits_{z\in\R} \hat\varepsilon_{N_M}(t, z).
\end{multline}
\par
When $j<M$ in \eqref{eq:B16}, it holds that 
\begin{multline}
\int_\R \delta_{x_0}(y) \int_{\R^M} \hat\varepsilon_{N_j}(t, z_j) \prod_{\substack{k=0\\k\neq j}}^M \Gamma_k(z_k) \d y \d \bm z
\\
\leq \int_\R \delta_{x_0}(y)\d y 
	\left(\int_{\R^M} \prod_{\substack{k=0\\k\neq j}}^{M-1} \Gamma_k(z_k)  
	\left(\int_\R \Gamma_M(z_M) \d z_j\right) \d \bm{z_*}\right)
	\sup\limits_{z\in\R} \hat\varepsilon_{N_j}(t, z)
\\
\leq \sup\limits_{z\in\R} \hat\varepsilon_{N_j}(t, z),
\end{multline}
{where $\bm{z_*}:= (z_0, \dots, z_{j-1}, z_{j+1}, \dots, z_{M-1})$. Here, it was used that} the integrand $\Gamma_M$ with argument $z_M = x - y - \sum_{j=0}^{M-1} z_j$ can be integrated against the `vacant' integrator $z_j$, which was freed up by bounding $\hat\varepsilon_{N_j}(t, z_j)$ against its supremum.
\par
From \eqref{eq:multibinomialsplit}, the integrals containing all the mixed terms $H(\bm z)$ with at least two error terms remain. We exemplary treat a summand with exactly two error terms, as the argument applies analogously to all other summands. 
To this end, consider
\begin{equation}
h_{ij} = \int_\R\int_{\R^M} \delta_{x_0}(y) \hat\varepsilon_{N_j}(t, z_j) \hat\varepsilon_{N_i}(t, z_i) \smashoperator{\prod_{\substack{k=0 \\ k\neq j, k \neq i}}^M} \Gamma_k(z_k) \d y \d \bm z,
\end{equation}
and let, without loss of generality, $i \neq M$, as we may otherwise exchange the roles of $i$ and $j$. Bounding $\hat\varepsilon_{N_j}(t, z_j)$ by its supremum, exactly like in the previous step, must necessarily produce an integral $\int_\R \hat\varepsilon_{N_i}(t, z_i) \d z_i = 0$, by  \Cref{prop:quaderrorintegral}. Existence of such a zero factor mandates that $h_{ij} = 0$. The same argument can be made for the remaining summands in $H(\bm z)$, whenever two or more error terms are present. 
Therefore, it holds that 
\begin{equation}
\int_\R\int_{\R^M} \delta_{x_0}(y) H(\bm z) \d y \d \bm z\ = 0,
\end{equation}
and we obtain
\begin{equation}
f\left(x; \X(t)\right) =  \left(\delta_{x_0} \ast \Gamma_0 \ast \cdots \ast \Gamma_M\right) (x) + \hat\varepsilon(t),
\end{equation}
with $\hat\varepsilon(t) \leq \sum_{j=0}^M \sup\limits_{y\in\R} \hat\varepsilon_{N_j}(t, y).$
\par \noindent
It remains to show that 
\begin{equation}\label{eq:convolutiondensi1}
\left(\delta_{x_0} \ast \Gamma_0 \ast \cdots \ast \Gamma_M\right) (x) 
= \sumquad \Wquad
f\left(x; x_0 + \sum_{j=0}^M \Yquad{n}{j}(s_j(t)) \right),
\end{equation}
since $f\left(x; x_0 + \sum_{j=0}^M \Yquad{n}{j}(s_j(t))\right) = f\left(x; \Xquad(t)\right)$ follows immediately from \Cref{def:Xprocess} of $\Xreal{}(t)$.
\par
With the distributive property and associativity of the convolution, it is possible to rewrite the convolution representation of the density as
\begin{equation*}
\left(\delta_{x_0} \ast \Gamma_0 \ast \cdots \ast \Gamma_M\right) (x) 
= \hspace{-8pt} \sumquad \Wquad \left( f_{\Yquad{n}{0}(s_0(t))} \ast \dots \ast f_{\Yquad{n}{M}(s_M(t))} \ast \delta_{x_0}	\right) (x).
\end{equation*}
{Here, we used that} for values of $j>\sup\{k\colon \tau_k \leq t\}$, it holds that $s_j(t) = 0$ and thus $f\left(x; \Y{j}(s_j(t))\right) = \delta_0(x)$. 
For such values of $j$, it holds the trivial quadrature representation
\begin{equation}
\delta_0(x) 
= \sum_{n_j=1}^{N_j} w_{n_j} f\left(x; \Yquad{n}{j}(0)\right) =  \sum_{n_j=1}^{N_j} w_{n_j} \delta_0(x),
\end{equation}
since the sum of quadrature weights is one, $\sum_{n_j=1}^{N_j} w_{n_j}=1$. The Dirac delta function $\delta_0$ is the identity element of convolution, therefore we need not specially accommodate for component processes which have not yet been reached by time $t$.

%
{Finally,} by the independence of the component processes $\Yreal{j}$ and $\Yreal{k}$ for $j\neq k$, the convolutions {above} can be equivalently expressed as the density of a sum,
\begin{equation}
\left( f_{\Yquad{n}{0}(s_0(t))} \ast \dots \ast f_{\Yquad{n}{M}(s_M(t))} \ast \delta_{x_0}\right) (x) = f\left(x; x_0 + \sum_{j=0}^M \Yquad{n}{j}(s_j(t))\right),
\end{equation}
which concludes the proof.
\end{proof}